\documentclass[journal,draftclsnofoot,onecolumn,12pt]{IEEEtran}
\ifCLASSINFOpdf
\else
\fi
\usepackage{amssymb,amsmath,amsfonts,amsthm}
\usepackage{mathrsfs}
\usepackage{cite}
\usepackage{array}
\usepackage{tabularx}
\usepackage[dvipsnames]{xcolor}
\usepackage{multirow}
\usepackage{adjustbox}
\usepackage{color}
\usepackage{booktabs}
\usepackage{xcolor,colortbl}
\definecolor{lavender}{rgb}{0.9, 0.9, 0.98}
\usepackage{mathtools}
\usepackage{subfigure}
\usepackage{mathtools}
\usepackage{tikz,pgf}
\usetikzlibrary{calc,positioning,mindmap,trees,decorations.pathreplacing}
\usepackage{hyperref}
\usepackage{graphicx}
\usepackage{bbm}
\usepackage[utf8]{inputenc}
\usepackage{xcolor}
\usepackage{algpseudocode}
\usepackage[linesnumbered,ruled,vlined]{algorithm2e}
\SetKwInput{KwInput}{Input}                
\SetKwInput{KwOutput}{Output}
\SetKwInput{KwInit}{Initialization}
\newtheorem{theorem}{Theorem}
\newtheorem{remark}{Remark}
\newtheorem{corollary}{Corollary}[theorem]

\begin{document}
%
\title{Statistically Optimal Beamforming and Ergodic Capacity for RIS-aided MISO  Systems}
%
%
%

\author{Kali Krishna Kota, M. S. S. Manasa,  
        Praful D. Mankar, and Harpreet S. Dhillon
\thanks{K. K. Kota, M. S. S. Manasa, and P. D. Mankar are with the Signal Processing and Communication Research Center, International Institute of Information Technology, Hyderabad 500032, India (e-mail: kali.kota@research.iiit.ac.in, mss.manasa@research.iiit.ac.in, praful.mankar@iiit.ac.in). H. S. Dhillon is with Wireless@VT, Department of ECE, Virginia Tech, Blacksburg, VA (Email:  hdhillon@vt.edu). This paper has been submitted in parts to IEEE GLOBECOM 2023 \cite{kota2022optimal}}}

%
%

\markboth{}%
{Shell \MakeLowercase{\textit{et al.}}: Bare Demo of IEEEtran.cls for IEEE Journals}
%



\maketitle

\vspace{-1.5cm}
\begin{abstract}
This paper focuses on optimal beamforming to maximize the mean signal-to-noise ratio (SNR) for a passive reconfigurable intelligent surface (RIS)-aided multiple-input single-output (MISO) downlink system. We consider that both the direct and indirect (through RIS) links to the user experience correlated Rician fading. 
The assumption of passive RIS imposes the unit modulus constraint, which makes the beamforming problem non-convex. To tackle this issue, we apply semidefinite relaxation (SDR) for obtaining the optimal phase-shift matrix and propose an iterative algorithm to obtain the fixed-point solution for statistically optimal transmit beamforming vector and RIS-phase shift matrix. Further, to measure the performance of the proposed beamforming scheme, we analyze key system performance metrics such as outage probability (OP) and ergodic capacity (EC). 
Just like the existing works, the OP and EC  evaluations rely on the numerical computation of the proposed iterative algorithm, which does not clearly reveal the functional dependence of system performance on key parameters such as line-of-sight (LoS) components, correlated fading, number of reflecting elements, number of antennas at the base station (BS), and fading factor. In order to overcome this limitation, we derive closed-form expressions for the optimal beamforming vector and phase shift matrix along with OP for special cases of the general setup. These expressions are then used to gain useful insights into the system performance and to understand the implications of the proposed solutions. Our analysis reveals that the independent and identically distributed ({\rm i.i.d.}) fading is more beneficial than the correlated case in the presence of LoS components. This fact is analytically established for the setting in which the LoS is blocked. Furthermore, we demonstrate that the maximum mean SNR improves linearly/quadratically with the number of RIS elements in the absence/presence of LoS component under {\rm i.i.d.} fading.

\end{abstract}
\vspace{-0.4cm}
\begin{IEEEkeywords}
\vspace{-0.2cm}
Reconfigurable Intelligent Surfaces, Optimal Beamforming, Statistical Beamforming, Spatially Correlated Channel, Outage Analysis, Ergodic Capacity.
\end{IEEEkeywords}
%
\IEEEpeerreviewmaketitle

\section{Introduction}
\IEEEPARstart{R}{econfigurable} intelligent surfaces (RIS) is a planar array that consists of many sub-wavelength-sized antenna elements formed of meta-materials. One can control the phases of these elements to change the way they interact with the impinging electromagnetic wave, thereby controlling the local propagation environment to a certain extent \cite{Emil_SigProcessing_2022}. This ability can influence key characteristics of the propagation environment, such as reflection, refraction, and scattering, which were thus far assumed to be uncontrollable in wireless communications systems \cite{CuiTieJun_CodingMetamaterial_2014,LiuYuanwei_RISprinciples_2021}. If configured properly, RISs can reduce the effects of fading and interference by redirecting the impinging signals such that they add constructively at the receiver. The advantages of such technology are low power consumption, high spectral efficiency, improved coverage, and better reliability \cite{PanCunhua_2021_6G,HuangChongwen_2019_EE}. Though RIS is conceptually similar to several existing technologies, such as relays, backscatter communication, and massive multiple-input multiple-output (MIMO), a key differentiating feature is its low-cost implementation and lower power consumption, especially when all the elements are passive (which will be our assumption). Additionally, RIS operates in the full-duplex mode as a passive device  without additional RF chains requirement and energy consumption, which is not the case in the competing technologies mentioned above ~\cite{WuQingging_RIS_comparisons,DiRenzo_RISvsRelay_2020}.

Due to its ability to create a smart propagation environment, RIS is envisioned as an enabling technology for the various applications in future wireless networks such as terahertz communications, simultaneous wireless information exchange, wireless power transfer, non-orthogonal multiple access, physical layer security, etc \cite{pradhan2022robust,PanCunhua_2021_6G,HuangJie_ChannelModeling_2022}. However, including RIS in communication systems comes with its own unique challenges. Particularly important ones are: 1) channel estimation and 2) optimal design of transmit/receive beamformers and RIS phase shift matrix.
The challenges in channel estimation stem from the need for estimation of the cascaded channels (i.e., transmitter-RIS and RIS-receiver) using the composite channel seen by the receiver \cite{ZhengBeixiong_ChEst_Survey} and the lack of active RIS elements to aid  channel estimation  \cite{MiguelDajer_2022}. 
On the other hand, the optimal selection of beamformer and phase shift matrix based on the given knowledge of the channel usually leads to the non-convex optimization formulation mainly because of the {\em unit modulus constraint} for the passive RIS \cite{BasarErtugrul_20219} and often because of the underlying objective function. This paper focuses on optimal beamforming to maximize the transmission capacity.

With little misuse of terminology, we will refer to the jointly optimal selection of the transmit beamformer for the BS and the phase shift matrix for the RIS elements as {\em optimal beamforming} for easier reference. The approaches to optimal beamforming for multi-antenna systems are primarily categorized based on the knowledge of channel state information (CSI) \cite{Goldsmith_2003}. By leveraging the perfect knowledge of  CSI (PSCI) at the transmitter, significant research efforts have been devoted to selecting instantaneously optimal beamforming for time-varying RIS-aided channels with a focus on optimizing a variety of key performance indicators, \cite{WuQingqing_2019_PSCI_TxPow,HuangChongwen_2018_SR,YuXianghao_MISO_PSCI_2019,YuXianghao_MISO_PSCI_2020,NingBoyu_RISMIMO_PSCI_2020}. The authors of \cite{WuQingqing_2019_PSCI_TxPow} applied semi-definite relaxation (SDR) to obtain optimal phase shift matrix that minimizes the transmission power for a single user and multi-user RIS-aided MIMO communication systems, whereas the authors of \cite{HuangChongwen_2018_SR} presented an alternating  majorization-minimization method based algorithm for optimal beamforming to maximize the sum rate. Further, the authors of \cite{YuXianghao_MISO_PSCI_2019,YuXianghao_MISO_PSCI_2020} proposed iterative algorithms such as fixed point iteration, manifold optimization, and branch-and-bound techniques to maximize the {\rm SNR}. Further, \cite{NingBoyu_RISMIMO_PSCI_2020 } proposed a new sum-path-gain maximization criterion to obtain a  suboptimal solution, which is numerically shown to achieve near-optimal RIS-MIMO channel capacity.
While the aforementioned works rely on PCSI, a few studies in the literature also perform instantaneous optimal beamforming for time-varying RIS-aided channel based on imperfect knowledge of CSI (IPCSI) to account  for the error associated with channel estimation. For example,  \cite{ZhiKangda_2022_IPCSI} optimized the sum rate for RIS-aided massive MIMO system with zero-forcing detectors and IPCSI. A projected gradient-based iterative algorithm is presented in \cite{NemanjaStefan_2021_RateOpt_IPCSI} to maximize the transmission rate for  multi-stream multiple-input RIS-aided MIMO systems with IPCSI.
While these works are performance achieving, they all considered instantaneous PCSI/IPCSI knowledge at the transmitter, which is not always feasible and practical, particularly for the RIS channel where the channel estimation itself is a cumbersome task. Moreover, as pointed out above, the formulations of optimal beamforming for RIS-aided systems are non-convex which are often solved through alternating subproblems (for transmit/receiver beamformer and phase shift matrix) based iterative algorithms. Additionally, finding solutions for the optimal phase shift matrix subproblem often relies on solvers like CVX/MOSEK. Because of this, the instantaneous optimal beamforming for RIS-aided systems becomes a computationally challenging task.
Furthermore, the instantaneous requirement of the PCSI/IPCSI feedback and the RIS reconfiguration increases the system complexity and requires complex RIS design for quick reconfigurability. Due to these reasons, it is practical to perform the optimal beamforming 
using the  statistical knowledge of CSI (SCSI) such that the system performance gets improved in a statistical sense, which is the main theme of this paper. 

In literature, there have been some efforts on the optimal beamforming for RIS-aided systems using the SCSI for a variety of channel models for the {\em direct link} from the BS to user equipment (UE)  and the {\em indirect link} from the BS to UE via RIS.  
The authors of \cite{hu2020statistical} considered a RIS-aided MISO system wherein the direct and indirect links  consist of LoS components along with multipath fading. To capture this, the wireless fading along indirect-direct links is modeled using  the Rician-Rician model.  We will refer to the fading along the indirect and direct links using a pair in a similar way throughout this paper.
Therein, the information of the angle of departure (AoD) and the angle of arrival (AoA) from/to BS/RIS defining the responses of LoS components  is appropriately used to obtain the  closed-form expressions of transmit beamformer and phase shift matrix that maximizes the mean SNR (which is equivalent to maximizing the upper bound on EC). The proposed solution requires  solving these closed-form expressions in an alternative fashion to arrive at a fixed-point optimal solution. The authors of \cite{HanYu_2019_SCSI} first obtained an upper bound on EC achieved with PCSI-based maximum ratio transmit beamforming for a large-scale RIS-aided MISO system under Rician-Rayleigh fading. Next, the authors derived a closed-form solution for a statistically optimal phase shift matrix of RIS that maximizes this upper bound on EC. In \cite{GanXu_2021_SCSI}, SCSI-based iterative optimal beamforming algorithms are presented to maximize the sum ergodic capacities for RIS-assisted multi-user MISO uplink and downlink  systems under the Rician-Rician fading scenario. For this scenario, the SCSI-based optimization formulation usually becomes difficult to handle.  Therefore, authors utilized  optimization methods like the alternating direction method of multipliers, fractional programming, and alternating optimization methods for obtaining sub-optimal solutions for power allocation and phase shift matrix.
The authors of \cite{ZhiKangda_2021_SCSI_SPAWC} presented  an approach wherein an IPSCI-based maximum ratio combining at the BS and a SCSI-based RIS-phase shifts configuration is adopted  for the multi-user uplink system under Rician-Rayleigh fading, whereas in  \cite{ZhiKangda_SSCI_mMIMO_2021}, SCSI-based near-optimal RIS-phase shifts are obtained using genetic algorithm when the BS utilizes PSCI-based MRC for multi-user uplink massive MIMO system.
In a similar direction, there exist few additional works on SCSI-based optimal beamforming for a variety of systems, including
multi-user multi-cell downlink system with the direct link absent \cite{LuoCaihong_Cellular_MISO_SSCI_2021}, multi-pair user exchanging information via RIS when the direct links are absent \cite{PengZhangjie_2021_SCSI}, and probabilistic technique  of discrete RIS phase-shit optimization  \cite{pradhan2023probabilistic}. 

In most of the aforementioned and other similar works, the key factor that enables tractable formulations is an assumption of {\rm i.i.d}  fading coefficients, often coupled with the absence of  LoS component along the direct link. However, ensuring independent fading in multi-antenna systems, particularly in RIS-aided systems, may not always be practical. This is especially true when  a large number of RIS elements are placed in a compact uniform planar array (UPA). Thus, it is crucial to consider  correlated fading for designing SCSI-based optimal beamforming for RIS-aided systems. The authors of \cite{Papazafeiropoulos_2022_SCSI_CorrFading} considered SCSI-based maximization of EC with respect to the RIS phase shifts for correlated Rayleigh channel while assuming that the direct link is absent. Further, the authors of \cite{WangJinghe_CorrelatedFading_2021} considered correlated Rician-Rayleigh fading for SCSI-based optimal beamforming to maximize EC of RIS-aided MIMO systems.  The authors proposed an iterative algorithm wherein the  optimal transmit beamforming vector and phase shift matrix are solved alternatively using SDR.
Most of the algorithms presented above rely on alternating between the transmit beamformer and phase shift matrix sub-problems to reach a fixed point solution similar to the PSCI-based approach. Additionally, the optimal phase shift solutions in each iteration often depend on the numerical optimization solver. However, these numerical solutions often hinder further analytical investigation, such as understanding the exact functional dependence of the optimal solutions on key system parameters. {\em Hence, it is equally, or perhaps even more, important to obtain closed-form optimal solutions, which will be the objective of this paper}.

SCSI-based optimal beamforming for RIS-aided systems circumvents the system design complexity issues and often provides comparable performance to the instantaneous PCSI-based beamforming (particularly when a strong LoS component is present). Therefore, to compare its performance with the PCSI-based schemes, it is crucial to characterize the performance of the statistical beamforming scheme analytically. This is another important reason for deriving a closed-form solution for optimal beamforming that lends analytical tractability in the performance analysis.
The performance of the beamformer scheme is usually characterized using the OP and EC. 
The closed-from expressions of OPs are derived for RIS-aided SISO systems under Rician fading with the direct link absent when the phase shift matrices are chosen based on IPCSI  in \cite{BaoTingnan_2022_Outage} and SCSI in \cite{XUPeng_AsymOutProb_2022}. Further, the authors of \cite{VanChien} derived the coverage probabilities for an arbitrary and statistically optimal phase shift matrix for the RIS-aided SISO system under Rayleigh-Rayleigh fading scenario. However, OP and EC analysis under jointly optimal  transmit/receive beamforming and RIS phase shift matrix for general RIS-aided multi-antenna systems is not investigated. In this paper, we focus on the statistically optimal beamforming and its performance analysis for various fading models.  

{\em Contributions:} This paper focuses on designing SCSI-based optimal beamforming for maximizing the mean SNR of a RIS-aided MISO system under a correlated Rician-Rician fading model incorporating both direct and indirect links with LoS component and correlated multipath fading. For this setting, we propose an iterative algorithm and also analyzed its OP and EC. In addition, we also derived closed-form expressions for the optimal beamforming and its OP under various special cases of the above generalized fading model. Using these derived expressions, we provide useful insights related to their performance comparisons. {\em To the best of our knowledge, this is the first paper to consider such a general model for SCSI-based beamforming.} The closest work dealing with a similar problem is \cite{WangJinghe_CorrelatedFading_2021}, but it ignores the LoS component along the direct link. In the following, we summarize the key contributions of our work.
\begin{enumerate}
\item An iterative algorithm is developed for optimal beamforming to maximize the mean SNR under correlated  Rician-Rician fading.
~For this beamforming scheme, OP is shown to closely follow square of Rice distribution, which is then utilized to determine EC. The parameters of the outage depend on the numerical evaluation of the proposed algorithm. 
\item Next, we derived computationally efficient closed-form expressions for optimal beamforming and their OPs for $\rm{i.i.d.}$ Rician-Rayleigh, correlated Rayleigh-Rayleigh and $\rm{i.i.d.}$ Rayleigh-Rayleigh. Furthermore, for $\rm{i.i.d.}$ Rician-Rician case, we maximized a carefully formulated lower bound of the mean SNR to arrive at the closed-form solutions.   
\item We have shown that the correlated fading outperforms the {\rm i.i.d.} case under Rayleigh-Rayleigh scenario. Next, we present a comparative analysis of {\rm i.i.d.} fading under Rician-Rician, Rician-Rayleigh, and Rayleigh-Rayleigh settings through analytical expressions. 
\item Key takeaways based on our numerical analysis are: 1) the SCSI-based optimal beamforming  performs better under {\rm i.i.d.} fading compared to the correlated scenario in the presence of  LoS component, 2) the direct link presence becomes insignificant when the number of reflecting elements is large, 3) the achievable capacity decreases with the increase in the difference between AoDs of LoS components belonging to direct and indirect links under Rician-Rician fading. 
\end{enumerate}
{\em Notations:} $a^*$ and $|a|$ represent the conjugate and absolute value of $a$. $\left\lVert \mathbf{a} \right\lVert$ and $\mathbf{a}_i$ are the norm and the $i$-th element of vector $\mathbf{a}$, whereas $\mathbf{A}^T$, $\mathbf{A}^H$, $\|\mathbf{A}\|_F$, ${\rm trace}(\mathbf{A})$,  $\mathbf{A}_{i,:}$, $\mathbf{A}_{:,i}$ and $\mathbf{A}_{ij}$ are the transpose, Hermitian, Frobenius norm, trace, $i$-th row, $i$-th column and $ij$-th element of the matrix $\mathbf{A}$, respectively. The notation $\mathbb{C}^{M\times N}$ is the set of  $M \times N$ complex matrices, ${\rm I_M}$ is $M\times M$ identity matrix and $\mathbf{1}_{\rm M}$ is a $M\times 1$ vector with unit elements. $\mathbf{v_A}$ and $\lambda_\mathbf{A}$ are the principal eigenvector and eigenvalue of  $\mathbf{A}$.  $\odot$ is the hadamard product, $\rm{diag}(\mathbf{a})$ is a diagonal matrix such that vector $\mathbf{a}$ forms its diagonal, and $\mathcal{CN}(\boldsymbol{\mu},\mathbf{K})$ denotes complex  Gaussian distribution with mean $\boldsymbol{\mu}$ and covariance matrix $\mathbf{K}$.
\vspace{-.3cm}\section{System Model}\vspace{-.2cm}

We consider a RIS-aided MISO communication system consisting of a BS with $M$ antennas, a RIS with $N$ passive antenna elements, and a single antenna UE. The BS can transmit information to the UE by jointly utilizing the direct link (BS-UE) and the indirect link (BS-RIS-UE). We consider a more practical setup wherein the LoS components and multi-path fading are present along both the direct and indirect links. To incorporate this, we model BS-UE, BS-RIS and RIS-UE channels using Rician fading. 
It is worth noting that the spacing between RIS elements might not ensure the independence of the fading coefficients, especially when arranging a large number of elements in a compact UPA. Therefore, it is crucial to consider correlated fading for designing the SCSI-based optimal beamforming. Our main objective in this paper is to jointly optimize the transmit beamforming vector $\mathbf{f}$ and the RIS phase shift matrix $\mathbf{\Phi}$ by leveraging the statistical knowledge of CSI for the above setup. 
\vspace{-.4cm}\subsection{Spatially Correlated Rician Channel Model}\label{channel model}\vspace{-.2cm}
 We model the BS-UE channel $\mathbf{g}$, BS-RIS channel $\mathbf{H}$, and RIS-UE channel $\mathbf{h}$  using the Rician fading with a factor $K$. Such channels can be expressed as the superposition of a deterministic LoS and spatially correlated random multipath components. The direct link under the correlated fading channel can be expressed as 
\begin{equation}
    \mathbf{g} = \kappa_l\mathbf{\Bar{g}} + \kappa_n\mathbf{\Tilde{g}},\label{directlink}
\end{equation}
where $\kappa_l = \sqrt{\frac{K}{1+K}}$, $\kappa_n = \sqrt{\frac{1}{1+K}}$, $\mathbf{\Tilde{g}}\sim\mathcal{C}\mathcal{N}(0,\mathbf{R}_{\rm BT})$ is the multipath component with covariance matrix $\mathbf{R}_{\rm BT}$ and $\mathbf{\Bar{g}}$ is the deterministic LoS component which is defined by the response of the uniform linear array (ULA) at the BS as $\mathbf{\Bar{g}} = \mathbf{a}_M(\theta_{\rm bd}^{\rm d})$ such that $\theta_{bd}^{\rm d}$ is  AoD along the direct link from BS. The response vector of ULA is given by
    $\mathbf{a}_M(\theta) = \frac{1}{\sqrt{M}}[1,e^{-j\frac{2 \pi \lambda}{d}\sin(\theta)},\ldots,e^{-j\frac{2 \pi \lambda}{d}(M-1)\sin(\theta)}]^T$, where $d$ is the distance between the antenna elements and $\lambda$ is the operating wavelength.
 Similarly, we express the RIS-UE link as 
\begin{equation}
    \mathbf{h} = \kappa_l\mathbf{\Bar{h}} + \kappa_n\mathbf{\Tilde{h}} ,\label{indirect link_h}
\end{equation}
where $\mathbf{\Tilde{h}}\sim\mathcal{C}\mathcal{N}(0,\mathbf{R}_{\rm RT})$ is the multipath component with covariance matrix $\mathbf{R}_{\rm RT}$ at the RIS transmit end, $\mathbf{\Bar{h}} = \mathbf{a}_N(\theta_{\rm rd})$ and  $\theta_{\rm rd}$ is the AoD from RIS.
Now, the BS-RIS link can be given as
 \begin{equation}
    \mathbf{H} = \kappa_l\mathbf{\Bar{H}} + \kappa_n\mathbf{\Tilde{H}} ,\label{indirect link_H}
\end{equation}
where $\mathbf{\Tilde{H}}\sim\mathcal{C}\mathcal{N}(0,\mathbf{R}_{\rm RR})$ is the multipath component modeled using \textit{double-sided spatial correlation} as $\mathbf{\Tilde{H}} = \mathbf{\Tilde{R}}_{\rm RR}\mathbf{\Tilde{H}}_{\rm W}\mathbf{\Tilde{R}}_{\rm BT}$ such that $\mathbf{R}_{\rm RR} = \mathbf{\Tilde{R}}_{\rm RR}\mathbf{\Tilde{R}}_{\rm RR}^T$, $\mathbf{R}_{\rm BT} = \mathbf{\Tilde{R}}_{\rm BT}\mathbf{\Tilde{R}}_{\rm BT}^T$ and $\mathbf{\Tilde{H}_W}\sim\mathcal{C}\mathcal{N}(0,I)$. $\mathbf{R}_{\rm RR}$ is the correlation matrix at the RIS receive end. $\mathbf{\Bar{H}} = \mathbf{a}_N(\theta_{\rm ra})\mathbf{a}^T_M(\theta_{\rm bd}^{\rm i})$ is the LoS component given by the response matrix of UPA such that $\theta_{\rm ra}$ is the AoA at  RIS and $\theta_{\rm bd}^{\rm i}$ is the AoD at  BS.

To model the covariance matrix for MIMO channel, one can apply the widely used  \emph{Kronecker Separable Model
} \cite{KSM} for capturing the pairwise correlation between the antenna elements. However, recently in \cite{RIS_Corr_Fad}, it is shown that the above-mentioned model is inaccurate for the RIS channel and use the UPA geometry of RIS to derive a new model wherein the correlation between two antenna elements is given by ${\rm sinc}\frac{2d}{\lambda}$. 
Using this, we model the fading covariance matrices for the BS and for the RIS on receiving and transmitting ends such that their $ij$-th element is
\begin{equation*}
    \mathbf{R}_{{\rm BT},{ij}} = {\rm sinc}\frac{2d_{ij}}{\lambda} \text{~~and~~}\mathbf{R}_{{\rm RT},{ij}} = \mathbf{R}_{{\rm RR},{ij}} = {\rm sinc}\frac{2r_{ij}}{\lambda},
\end{equation*}
where $d_{ij}$ and $r_{ij}$ are the distances between the $i$-th and $j$-th antennas at BS and antenna elements at RIS, respectively.
\vspace{-.4cm}\subsection{Received Signal Model}\vspace{-.2cm}
Given the BS transmits symbol $x$, the signal received at UE is given by 
\begin{equation}
    y = l(d_1,d_2)\mathbf{h}^T\mathbf{\Phi Hf}x + l(d_0)\mathbf{g}^T\mathbf{f}x +n,\label{RxSig}
\end{equation}
where $n \sim \mathcal{N}(0,\sigma_n^2)$ is the complex Gaussian noise, $\mathbf{f}\in\mathbb{C}^N$ is the transmit beamformer, $\mathbf{\Phi}={\rm diag}(\boldsymbol{\psi})$ is the RIS phase shift matrix, and $l(d_1,d_2)$ and $l(d_0)$ are the far field path loss functions for the indirect and direct links, respectively. The transmit power constraint at the BS for $\mathbb{E}[xx^H] = P_s$ implies $\|\mathbf{f}\|^2 = 1$ where $P_s$ is the total transmission power. Further, the consideration of passive elements for RIS implies that the entries of $\boldsymbol{\psi}$ are complex with unit magnitude, i.e., $|\boldsymbol{\psi}_k| = 1 ~~ \forall k = 0,\ldots, N-1$. Also, in a far field scenario, the path loss for the RIS channel follows the ``product of distances" model such that $l(d_1,d_2) = (d_1d_2)^{-\alpha/2}$ \cite{LiuYuanwei_RISprinciples_2021} and BS-UE channel path loss follows $l(d_0) = (d_0)^{-\alpha/2}$, where $d_0$, $d_1$ and $d_2$ are the distances of BS-UE, BS-RIS, and RIS-UE links, respectively, and $\alpha$ is the path loss exponent.

Using \eqref{RxSig}, SNR can be written as
\begin{equation}
    \text{SNR} = \gamma |\mathbf{h}^T\mathbf{\Phi Hf} + \mu \mathbf{g}^T\mathbf{f}|^2,\label{SNR}
\end{equation}
where $\gamma = (d_1d_2)^{-\alpha}\frac{P_s}{\sigma_n^2} $ and $\mu = (\frac{d_0}{d_1d_2})^{-\alpha/2}$ is the path loss ratio (PLR) of direct and indirect links. It may be noted that the PLR captures the strength of the direct link compared to the indirect link and thus it will be an important parameter for the optimal beamforming. 
\vspace{-.4cm}\subsection{Problem Formulation}\label{prob from}\vspace{-.2cm}
This paper aims to maximize EC by jointly optimizing the transmit beamforming vector $\mathbf{f}$  and  RIS-phase shift matrix $\mathbf{\Phi}$. For given $\mathbf{f}$ and $\mathbf{\Phi}$, EC  is 
\begin{equation*}
    {\rm C} = \mathbb{E}[\log_2(1 + \gamma|\mathbf{h}^T\mathbf{\Phi Hf} + \mu \mathbf{g}^T\mathbf{f}|^2)].
\end{equation*}
However, the expectation of the log function is difficult to handle in the maximization problem. Thus, we apply Jensen's inequality 
and focus on maximizing the upper bound of capacity as 
\begin{align}
  {\rm C}\leq{\rm C_{ub}} &=  \log_2(1 + \Gamma(\mathbf{f},\mathbf{\Phi})),\text{~~where~~}
  \Gamma(\mathbf{f},\mathbf{\Phi}) = \gamma\mathbb{E}[|\mathbf{h}^T\mathbf{\Phi Hf} + \mu \mathbf{g}^T\mathbf{f}|^2]\label{obj_fun}
\end{align}
represents the mean {\rm SNR}. Henceforth, we will assume $\gamma=1$ without any loss of generality. Thus, the capacity maximization problem can be reformulated using its upper bound as
\begin{subequations}
\begin{align}
    \max_{\mathbf{f},\mathbf{\Phi}} ~~&  \Gamma(\mathbf{f},\mathbf{\Phi}),\label{objective}\\
   \text{s.t.} ~~& \|\mathbf{f}\|^2 = 1,\label{constraint_f}\\
    &|\boldsymbol{\psi}_k| = 1, ~~ \forall k = 0,\ldots,N-1.\label{constraint_phi}%
\end{align}
\label{optimization problem}%
\end{subequations}
\noindent where \eqref{constraint_f} is the unit norm constraint of the transmit beamformer and \eqref{constraint_phi} is the unit magnitude constraint on the passive RIS elements to ensure phase shifts without  amplification/attenuation.
The problem formulation given in \eqref{optimization problem} is non-convex due to \eqref{constraint_phi} and the objective function given in \eqref{obj_fun}. Further, the problem is coupled in terms of $\mathbf{f}$ and $\mathbf{\Phi}$ which makes it further difficult to solve the problem. Hence, to tackle this issue, we solve the optimization problem using alternating subproblems and provide iterative algorithms or closed-form solutions for optimal $\mathbf{f}$ and $\mathbf{\Phi}$ for various fading scenarios in section \ref{optimal beamforming}. Besides, we also decouple $\mathbf{f}$ and $\mathbf{\Phi}$ in some special cases of fading which allows us to get a closed-form solution $\mathbf{f}$ and $\mathbf{\Phi}$. 

In addition, we also analyze the \emph{outage probability} and \emph{ergodic capacity} to characterize the performance of the proposed SCSI-based beamforming schemes.
For a beamforming scheme with optimal beamformer $\mathbf{f}_{\rm opt}$ and phase shift matrix $\mathbf{\Phi}_{\rm opt}$, the OP and EC are given by
\begin{align}
    {\rm P_{out}}(\beta) &= \mathbb{P}[{\Gamma(\mathbf{f}_{\rm opt},\mathbf{\Phi}_{\rm opt})} \leq \beta]~~\text{and}\label{P_out}\\
      {\rm C} = \mathbb{E}[\log_2(1+\Gamma&(\mathbf{f}_{\rm opt},\mathbf{\Phi}_{\rm opt}))]=\frac{1}{\ln(2)}\int_0^\infty \frac{1}{1+u}\left(1-{\rm P_{out}}(u)\right){\rm d}u,\label{ergodic capacity}
\end{align}
respectively. In the next section, we present the algorithms/solutions for optimal beamforming and also perform outage analysis for variants of the channel model discussed in section \ref{channel model}.
\vspace{-.4cm}\section{Statistically Optimal Beamforming for RIS-aided systems }\label{optimal beamforming}\vspace{-.2cm}
In  our system, the RIS is positioned such that it has an LoS path with both the BS and the UE. Hence, Rician distribution is adopted to model the indirect links $\mathbf{H}$ and $\mathbf{h}$. Besides, we also consider an LoS Path at the direct link which is ignored in \cite{WangJinghe_CorrelatedFading_2021}. This makes the considered RIS-aided MISO system more general. 
For such a system, our objective is to maximize the upper bound on EC as discussed in Section \ref{prob from}. The objective function of problem \eqref{optimization problem}, i.e. the mean {\rm SNR}, can be written as
\begin{align}
   \Gamma(\mathbf{f},\mathbf{\Phi}) = |\kappa_l^2\boldsymbol{\psi}^T\mathbf{Ef} + \mu\kappa_l\mathbf{\Bar{g}}^T\mathbf{f}|^2 + \kappa_l^2\kappa_n^2\boldsymbol{\psi}^H\mathbf{Z}_1\boldsymbol{\psi} + \kappa_n^2[\mu^2+\boldsymbol{\psi}^H\mathbf{Z}_2\boldsymbol{\psi}]\mathbf{f}^H\mathbf{R}_{\rm BT}\mathbf{f},\label{eq:Mean_SNR_R1}
\end{align}
where $\mathbf{E}={\rm diag}(\mathbf{\Bar{h}})\mathbf{\Bar{H}}$, $\mathbf{Z}_1 = \mathbf{R}_{\rm RT}\odot\mathbf{\Bar{H}ff}^H\mathbf{\Bar{H}}^H$ and $\mathbf{Z}_2 = \mathbf{R}_{\rm RR}\odot( \kappa_n^2\mathbf{R}_{\rm RT}+\kappa_l^2\mathbf{\Bar{h}}^*\mathbf{\Bar{h}}^T)$. The proof of \eqref{eq:Mean_SNR_R1} is given in Appendix \ref{AppA}. As mentioned earlier, the non-convex nature of the problem makes it challenging to directly obtain optimal $\mathbf{f}$ and $\mathbf{\Phi}$. Therefore, we tackle this issue by dividing the problem into optimal beamformer and phase shift matrix sub-problems as follows\\
\emph{1) Optimal Beamformer}: For a given phase shift matrix $\mathbf{\Phi}$, the optimization problem with respect to the beamforming vector $\mathbf{f}$ becomes
\begin{subequations}
\begin{align}
    \max_{\mathbf{f}} ~~& \mathbf{f}^H\mathbf{F}\mathbf{f}, \label{subprob_f_R1}\\
\text{s.t.} ~~& \|\mathbf{f}\|^2 = 1,
\end{align}
\end{subequations}
where the objective function follows from \eqref{eq:Mean_SNR_R1} with
\begin{equation}
    \mathbf{F} = \mathbf{F}_1 + \mathbf{F}_2 + \mathbf{F}_3,\label{M}
\end{equation}
such that $\mathbf{F}_1= (\kappa_l^2\boldsymbol{\psi}^T\mathbf{E} + \mu\kappa_l\mathbf{\Bar{g}}^T)^H(\kappa_l^2\boldsymbol{\psi}^T\mathbf{E} + \mu\kappa_l\mathbf{\Bar{g}}^T)$, $\mathbf{F}_2=\kappa_l^2\kappa_n^2\mathbf{\Bar{H}}^H\mathbf{\Phi}^H\mathbf{R}_{\rm RT}\mathbf{\Phi\Bar{H}}$, and $\mathbf{F}_3=\kappa_n^2\mathbf{R}_{\rm BT}[\mu^2+\boldsymbol{\psi}^H\mathbf{Z}_2\boldsymbol{\psi}]$.
It is to be noted that  $\mathbf{F}_1$, $\mathbf{F}_2$, and $\mathbf{F}_3$ are symmetric matrices which implies that $\mathbf{F}$ is also a symmetric matrix. Thus, this optimization problem is equivalent to the Rayleigh quotient maximization, whose solution, i.e., the optimal transmit beamformer, becomes the dominant eigenvector of $\mathbf{F}$ and can be given as 
    \begin{equation}
        \mathbf{f}_{\rm opt} = \mathbf{v_F},\label{OptSol_f_R1}
    \end{equation}
where $\mathbf{v_F}$ is the dominant eigenvector of $\mathbf{F}$.\\
\emph{2) Optimal Phase Shift Matrix}: For a given beamforming vector $\mathbf{f}$, the optimization problem with respect to the phase shift matrix $\mathbf{\Phi} = \rm{diag}(\boldsymbol{\psi})$ becomes
\begin{subequations}
\begin{align}
    \max_{\boldsymbol{\psi}}  ~~& |\boldsymbol{\psi}^H\mathbf{a}+\mathbf{b}|^2 + \boldsymbol{\psi}^H\mathbf{V}\boldsymbol{\psi} \label{subprob_phi_R1},\\
   \text{s.t.} ~~& |\boldsymbol{\psi}_k| = 1 ~~ \forall k = 0,\ldots,N-1,\label{constraint_psi}
\end{align}\label{prob psi}%
\end{subequations}
where the objective function follows by rewriting \eqref{eq:Mean_SNR_R1} with $\mathbf{V}=\kappa_l^2\kappa_n^2\mathbf{Z}_1+\kappa_n^2\mathbf{f}^H\mathbf{R}_{\rm BT}\mathbf{f}\mathbf{Z}_2$, $ \mathbf{a}=\kappa_l^2\mathbf{Ef}$, and $\mathbf{b}= \mu\kappa_l\mathbf{\Bar{g}}^T\mathbf{f}$. 
Since the above problem is non-convex, we model  \eqref{prob psi} as a semidefinite programming problem by introducing an auxiliary variable as below
\begin{subequations}
\begin{align}
    \max_{\Bar{\boldsymbol{\psi}}} ~~& \Bar{\boldsymbol{\psi}}^H\mathbf{A}\Bar{\boldsymbol{\psi}} + \|\mathbf{b}\|^2,\label{objective_psibar}\\
   \text{s.t.} ~~& |\Bar{\boldsymbol{\psi}_k}| = 1 ~~ \forall k = 0,\ldots,N-1,\label{constraint_psibar}
\end{align}\label{obj psibar}%
\end{subequations}
where $\mathbf{A} = \begin{bmatrix}
                 \mathbf{aa}^H+\mathbf{V} & \mathbf{ab}^H\\
                 \mathbf{ba}^H & 0\label{MatrixA}
                 \end{bmatrix}$ and $\Bar{\boldsymbol{\psi}} = [ \boldsymbol{\psi}~~t]^T$.
Further, defining $\mathbf{\Psi} = \Bar{\boldsymbol{\psi}}\Bar{\boldsymbol{\psi}}^H$ such that ${\rm diag}(\mathbf{\Psi}) = 1$ will ensure the constraint in \eqref{constraint_psibar} and also 
allow us to rewrite \eqref{objective_psibar} as ${\rm tr}(\mathbf{A\Psi})$. 
Next, we use the standard SDR technique to solve the above problem as below  
\begin{subequations}
\begin{align}
    \max_{\mathbf{\Psi}} ~~& {\rm tr}(\mathbf{A\Psi}) ,\\
   \text{s.t.} ~~& \mathbf{\Psi} \succeq 0, \hspace{0.2cm}{\rm diag}(\mathbf{\Psi}) = 1.
\end{align}\label{relaxed_phi}%
\end{subequations}
This problem can be solved by using standard solvers such as CVX/Mosek. Finally, the optimization problem given in \eqref{optimization problem} can now be solved by iterating over the beamforming and phase shift matrix subproblems using \eqref{subprob_f_R1} and \eqref{relaxed_phi}, respectively, as summarized in Algorithm \ref{Alg1}.
\begin{algorithm}\label{Alg1}
\SetKwComment{Comment}{$\triangleright$\ }{}
\KwInput{$\mu$, $\kappa_l$, $\kappa_n$, $\mathbf{\Bar{h}}$, $\mathbf{\Bar{H}}$, $\mathbf{\Bar{g}}$, $\mathbf{R}_{\rm RT}$, $\mathbf{R}_{\rm RR}$, $\mathbf{R}_{\rm BT}$, $\delta$}
  \KwInit{$\mathbf{f_0}$ , $\boldsymbol{\psi_0}$}
\SetKwRepeat{Repeat}{Repeat}{Untill:}
\Repeat{$|\Gamma(\mathbf{f_{i}},\mathbf{\Phi_{i}})-\Gamma(\mathbf{f_{i-1}},\mathbf{\Phi_{i-1}})| \leq \delta$}{Set $\mathbf{f}=\mathbf{f_{i-1}}$ and evaluate $\mathbf{A}$ using \eqref{MatrixA}.\\
  Using $\mathbf{A}$, solve  \eqref{relaxed_phi} for $\mathbf{\Psi}$. 
  Obtain $\boldsymbol{\psi_i}$ such that $\left[ \boldsymbol{\psi_i~}  t\right]^T\left[ \boldsymbol{\psi_i}^* ~ t^*\right]=\mathbf{\Psi}.$\\
  Set $\boldsymbol{\psi}=\boldsymbol{\psi_i}$ and evaluate $\mathbf{F}$ using \eqref{M}.\\
  Obtain $\mathbf{f_{i}}$ using \eqref{OptSol_f_R1} such that $ \mathbf{f_i} = \mathbf{v_F}$. \\
  $\mathbf{i} \gets \mathbf{i}+1$.}
\caption{SCSI-based optimal beamforming for R1 correlated fading}
\end{algorithm}
For the proposed beamforming scheme above, we now present OP and EC. In particular, we analyze the outage for a given $\mathbf{f}$ and $\mathbf{\Phi}$ due to not having closed-form expressions.
Defining  $\xi_1 = \mathbf{h}^T\mathbf{\Phi Hf}$ and $\xi_2 =  \mathbf{g}^T\mathbf{f}$ will allow to rewrite the OP given in \eqref{P_out}  as
\begin{equation*}
    \text{P}_{\rm out}(\beta) = \mathbb{P}\left[|\xi_1 + \mu \xi_2| \leq \sqrt{\beta/\gamma}\right].
\end{equation*}
 In  Appendix \ref{AppB}, we show that $\xi_1$ (closely) and $\xi_2$ follow complex Gaussian distributions as 
\begin{equation}
\xi_1\sim\mathcal{C}\mathcal{N}(\mu_1,\sigma_1^2)\hspace{0.2cm}\text{and}\hspace{0.2cm} 
    \xi_2\sim\mathcal{C}\mathcal{N}(\mu_2,\sigma_2^2),\label{eq:xi1xi2}
\end{equation}
\begin{align*}
\text{where}~~ \mu_1 =\kappa_l^2\mathbf{\Bar{h}}^T\mathbf{\Phi \bar{H}f} &\text{~and~}\sigma_1^2 =\kappa_l^2\kappa_n^2\boldsymbol{\psi}^H\mathbf{Z}_1\boldsymbol{\psi} + \kappa_n^2\mathbf{f}^H\mathbf{R}_{\rm BT}\mathbf{f}[\boldsymbol{\psi}^H\mathbf{Z}_2\boldsymbol{\psi}],\\ 
    \mu_2 = \kappa_l\mathbf{\Bar{g}}^T\mathbf{f} &\text{~and~}\sigma_2^2 =\kappa_n^2\mathbf{f}^H\mathbf{R}_{\rm BT}\mathbf{f}.
\end{align*} 
Using the independence of $\xi_1$ and $\xi_2$, we get
\begin{equation}
    \xi_1 + \mu\xi_2\sim\mathcal{C}\mathcal{N}(m,\sigma^2)
~~\text{where}~~m = \mu_1+\mu\mu_2,
  \text{~~and~~}  \sigma^2 = \sigma_1^2+\mu^2\sigma_2^2,  \label{parameters_m_sigma}
\end{equation} 
are the mean and variance of $\xi_1 + \mu \xi_2$, respectively.
Further, we also use the fact that the magnitude of a non-zero mean complex Gaussian follows the Rice distribution to obtain
\begin{equation}
    |\xi_1 + \mu\xi_2|\sim{\rm Rice}\left(|m|,\sigma\right).\label{SNR_dist}
\end{equation}
From \eqref{SNR_dist}, we can write OP as given in the following theorem using the  CDF of the Rice distribution with parameters $m$ and $\sigma^2$ given in \eqref{parameters_m_sigma} which are evaluated  using the optimal beamformer $\mathbf{f}_{\rm opt}$ and phase shift matrix $\mathbf{\Phi}_{\rm opt}$ obtained through Algorithm \ref{Alg1}. 
\begin{theorem}\label{Theo1}
OP of the SCSI-based optimal beamforming scheme for the RIS-aided MISO system under correlated Rician-Rician fading is given by
\begin{align}
    {\rm P_{out}}(\beta)\approx 1-Q_1\left(\frac{|m|}{\sqrt{\sigma/2}},\frac{\sqrt{\beta/\gamma}}{\sqrt{\sigma/2}}\right),\label{Pout_R1}
\end{align}
where $m =  \kappa_l^2\mathbf{\Bar{h}}^T\mathbf{\Phi_{\rm opt} \mathbf{\Bar{H}} f_{\rm opt}} + \mu \kappa_l\mathbf{\Bar{g}}^T\mathbf{f}_{\rm opt}$, $\sigma^2=  \kappa_l^2\kappa_n^2\boldsymbol{\psi}_{\rm opt}^H\mathbf{Z}_1\boldsymbol{\psi}_{\rm opt} + \kappa_n^2\mathbf{f}_{\rm opt}^H\mathbf{R}_{\rm BT}\mathbf{f}_{\rm opt}[\boldsymbol{\psi}_{\rm opt}^H\mathbf{Z}_2\boldsymbol{\psi}_{\rm opt}]+ \mu^2 \kappa_n^2\mathbf{f}_{\rm opt}^H\mathbf{R}_{\rm BT}\mathbf{f}_{\rm opt}$,  
 $\mathbf{f}_{\rm opt}$ and $\mathbf{\Phi}_{\rm opt}$ are solutions of Algorithm \ref{Alg1}, and  $Q_1(\cdot)$ is a Marcum Q-function.
\end{theorem}
Using Theorem \ref{Theo1} and \eqref{ergodic capacity}, we determine EC of the proposed beamforming scheme  in Algorithm \ref{Alg1} in the following corollary.
\begin{corollary}
\label{cor:R1_EC}
EC of the SCSI-based optimal beamforming scheme for the RIS-aided MISO system under correlated Rician-Rician fading is given by
\begin{align}
    {\rm C} \approx \frac{1}{\ln(2)}\int_0^\infty \frac{1}{1+u}Q_1\left(\frac{|m|}{\sqrt{\sigma/2}},\frac{\sqrt{{u}/{\gamma}}}{\sqrt{\sigma/2}}\right){\rm d}u.
\end{align}
where $m$ and $\sigma^2$ are given in Theorem \eqref{Theo1}.
\end{corollary}

Now, we will study the impact of the limiting cases of fading factor $K$ on the SCSI-based optimal  beamforming  and its performance.   \newline
{\em 1) Case $K\to\infty$:} The multipath fading vanishes  as $K$ becomes large for which the resulting channels are described by their deterministic LoS components such that $\mathbf{g}=\mathbf{\Bar{g}}$, $\mathbf{h} = \mathbf{\Bar{h}}$, and $\mathbf{H} = \mathbf{\Bar{H}}$.  For this case, the  {\rm SNR} reduces to a deterministic value for which  \eqref{obj_fun} becomes 
\begin{equation}
\Gamma(\mathbf{f},\mathbf{\Phi}) = |\mathbf{\Bar{h}}^T\mathbf{\Phi\Bar{H}f} + \mu\mathbf{\Bar{g}}^T\mathbf{f}|^2 = |\boldsymbol{\psi}^T\mathbf{Ef} + \mu\mathbf{\Bar{g}}^T\mathbf{f}|^2.\label{eq:R1_Cor_K_inftt}
\end{equation}
 To maximize \eqref{eq:R1_Cor_K_inftt} with constraints \eqref{constraint_f} and \eqref{constraint_phi}, we can easily find
 \begin{align}
     \mathbf{f}_{\rm opt}=\frac{\mathbf{E}^H\boldsymbol{\psi}^*_{\rm opt} + \mu\mathbf{\Bar{g}}^*}{\|\mathbf{E}^H\boldsymbol{\psi}^*_{\rm opt} + \mu\mathbf{\Bar{g}}^*\|}\text{~~and~~}\boldsymbol{\psi}_{\rm opt}=\exp\left(-j\angle{\mathbf{E}\mathbf{f}_{\rm opt}} + j\angle{\mathbf{\Bar{g}}^T\mathbf{f}_{\rm opt}}\right).\label{eq:R1_K_infty}
 \end{align}
 As the optimal $\mathbf{f}_{\rm opt}$ and $\boldsymbol{\psi}_{\rm opt}$ given in \eqref{eq:R1_K_infty} depend on each other, they can be solved alternatively for a fixed point solution.
 \newline {\em 2) Case $K\to 0$:} In this case, the LoS components become insignificant and the channels get completely characterized   by the multipath fading  such that $\mathbf{g}=\mathbf{\tilde{g}}$, $\mathbf{h} = \mathbf{\tilde{h}}$, and $\mathbf{H} = \mathbf{\tilde{H}}$. This is equivalent to the correlated Rayleigh-Rayleigh fading model for which the optimal beamforming and its performance will be presented in  Section \ref{R3}.

The outage and capacity performances given in Theorem \ref{Theo1} and Corollary \ref{cor:R1_EC} allow us to numerically evaluate the system performance as they rely on $\mathbf{f}_{\rm opt}$ and $\mathbf{\Phi}_{\rm opt}$ which are obtained using Algorithm \ref{Alg1}.  Thus, though they are expressed in a simple analytical form, their evaluation is limited by a computationally complex algorithm. Such numerical solutions may not always provide useful insights into the exact functional dependence of the optimal solution on the key system parameters.
Thus, it is desirable to have a closed-form expression for the optimal beamformer $\mathbf{f}$ and $\mathbf{\Phi}$ so that the outage and capacity performances can be characterized analytically without relying on the numerical evaluation of an algorithm. Additionally, having closed-form expressions for optimal $\mathbf{f}$ and $\mathbf{\Phi}$ will help to reduce the implementation complexity. Motivated by this, we investigate special cases of the channel model given in Section \ref{channel model} for deriving closed-form expressions or computationally efficient solutions  in the following subsections.
\vspace{-.5cm}
\subsection{IID Rician-Rician Fading}\label{R1_IID}
In this subsection, we focus on solving the optimization problem in \eqref{optimization problem} while considering an {\rm i.i.d.} Rician-Rician fading model for direct and indirect links which implies that $\mathbf{\Tilde{g}}\sim\mathcal{C}\mathcal{N}(0,\mathbf{I_M})$, $\mathbf{\Tilde{h}}\sim\mathcal{C}\mathcal{N}(0,\mathbf{I_N})$ and $\mathbf{\Tilde{H}}_{:,i}\sim\mathcal{C}\mathcal{N}(0,\mathbf{I_N})$. For this case, by simplifying the steps given in Appendix \ref{AppA} with identity covariance matrices, we can write the  mean {\rm SNR} given in \eqref{obj_fun} as
\begin{equation}
    \Gamma(\mathbf{f},\mathbf{\Phi}) = |\kappa_l^2\mathbf{\Bar{h}}^T\mathbf{\Phi\Bar{H}f} + \mu\kappa_l\mathbf{\Bar{g}}^T\mathbf{f}|^2 + \kappa_l^2\kappa_n^2\|\mathbf{\Bar{H}f}\|^2 +  (\kappa_l^2\kappa_n^2+\kappa_n^4)N + \mu^2\kappa_n^2.  \label{OptProb_R1_IID}
\end{equation}
 Maximization of mean SNR given in \eqref{OptProb_R1_IID}  has already been addressed by the authors of \cite{hu2020statistical}. 
However, they proposed an alternating optimization-based approach as given in Algorithm \ref{Alg2}. While  Algorithm \ref{Alg2} is computationally efficient compared Algorithm \ref{Alg1}, the system performance analysis is still limited by the numerical evaluation of  the optimal beamforming. 
\begin{algorithm}[t!]\label{Alg2}
\caption{SCSI-based optimal beamforming for IID R1  case \cite{hu2020statistical}}
\SetKwComment{Comment}{$\triangleright$\ }{}
\KwInput{$\mu$, $\kappa_l$, $\kappa_n$, $\mathbf{\Bar{h}}$, $\mathbf{\Bar{H}}$, $\mathbf{\Bar{g}}$, $\delta$}
\KwOutput{$\mathbf{f}_{\rm opt}$ , $\boldsymbol{\psi_{\rm opt}}$}
  \KwInit{$\mathbf{f}_0$ , $\boldsymbol{\psi}_0$}
\SetKwRepeat{Repeat}{Repeat}{Untill:}
\Repeat{$|\Gamma(\mathbf{f}_{i},\mathbf{\Phi}_{i})-\Gamma(\mathbf{f}_{i-1},\mathbf{\Phi}_{i-1})| \leq \delta$}{$\mathbf{Y}=\begin{bmatrix}
      \kappa_l^2\boldsymbol{\psi}_{i-1}^T{\rm diag}(\mathbf{\Bar{h}})\mathbf{\Bar{H}} + \mu\kappa_l\mathbf{\Bar{g}}^T~~
      \kappa_l\kappa_n\mathbf{\Bar{H}}
  \end{bmatrix}^T$ and set $\mathbf{f}_{i}=\mathbf{v_Y}$.\\
  Set $\boldsymbol{\psi}_{i}=\exp\{-j\angle{\left({\rm diag}(\mathbf{\bar{h}})\mathbf{\bar{H}}\mathbf{f}_{i} - \mathbf{\bar{g}}^T\mathbf{f}_{i}\right)}\}$.\\
$\mathbf{i} \gets \mathbf{i}+1$.}
\end{algorithm}
Thus, to obtain the closed form solution, we  maximize a carefully constructed lower bound of  the mean SNR.  For this, we first optimize it w.r.t. $\mathbf{\Phi}$ as below. \\
\emph{1) Optimal Phase Shift Matrix}: Using the mean SNR given in \eqref{OptProb_R1_IID}, the optimization problem \eqref{optimization problem} with respect to the phase shift matrix $\mathbf{\Phi}$ for a given $\mathbf{f}$ becomes 
\begin{subequations}
\begin{align}
    \max_{\boldsymbol{\psi}}~~ &|\kappa_l^2 \boldsymbol{\psi}^T{\rm diag}(\mathbf{\bar{h}})\Bar{\mathbf{H}}\mathbf{f} +\mu \kappa_l \bar{\mathbf{g}}^T\mathbf{f}|^2\label{subprob_Phi_R1_IID} \\
    {\rm\text{s.t.}}~~&|\boldsymbol{\psi}_k|=1, ~~\forall k=0,\cdots,N-1,\label{constraint_Phi_R1_IID}
\end{align}
\end{subequations} 
where \eqref{subprob_Phi_R1_IID} follows from \eqref{OptProb_R1_IID}. 
The above objective function can be upper bounded as
\begin{align*}
     |\kappa_l^2 \boldsymbol{\psi}^T\rm{diag}(\bar{\mathbf{h}}) \bar{\mathbf{H}} \mathbf{f} +\mu \kappa_l \bar{\mathbf{g}}^T \mathbf{f}|^2 \leq  \kappa_l^4 |\boldsymbol{\psi}^T{\rm diag}(\mathbf{\bar{h}}) \bar{\mathbf{H}} \mathbf{f}|^2 +\mu^2 \kappa_l^2 |\bar{\mathbf{g}}^T\mathbf{f}|^2,
\end{align*}
where equality holds when $\boldsymbol{\psi}^T{\rm diag}(\mathbf{\bar{h}})\bar{\mathbf{H}}\mathbf{f} = c  \bar{\mathbf{g}}^T\mathbf{f}$ for a constant $c$. To achieve  equality, we set 
\begin{align}
 \boldsymbol{\psi}^T=c  \bar{\mathbf{g}}^T\mathbf{f} \mathbf{w},\label{psi^t} 
\end{align}
where $\mathbf{w}=\mathbf{f}^H\mathbf{E}^H/\|\mathbf{Ef}\|^2$ is the pseudoinverse of $\mathbf{Ef}$ such that $\mathbf{E}={\rm diag}(\mathbf{\bar{h}})\mathbf{\Bar{H}}$. However, the constraint in \eqref{constraint_Phi_R1_IID} needs to be satisfied. Interestingly, as  will be clear shortly, we have $|\mathbf{w}_i|=|\mathbf{w}_j|$ $\forall i,j$. Thus, we can set $c$ such that  \eqref{constraint_Phi_R1_IID} is ensured. 
Let $\mathbf{e}_n=\mathbf{E}_{n,:}$ be the $n$-th row of $\mathbf{E}$.
By  construction of $\bar{\mathbf{h}}$ and $\bar{\mathbf{H}}$ (see Section \ref{channel model}), we have $\mathbf{e}_n^H\mathbf{e}_n=\mathbf{e}_m^H\mathbf{e}_m$ using which we get
 \begin{align}
 \mathbf{f}^H\mathbf{e}_n^H\mathbf{e}_n\mathbf{f}=\mathbf{f}^H\mathbf{e}_m^H\mathbf{e}_m\mathbf{f} \Rightarrow \|\mathbf{e}_n\mathbf{f}\|^2=\|\mathbf{e}_m\mathbf{f}\|^2.\label{emf}
 \end{align}
 In addition, we also observe that 
\begin{align}
  \|\mathbf{E}\mathbf{f}\|^2=\sum\nolimits_{n=0}^{N-1}|\mathbf{e}_n\mathbf{f}|^2=N|\mathbf{e}_n\mathbf{f}|^2.\label{Ef}
\end{align}
Using \eqref{emf} and \eqref{Ef}, we can write 
\begin{equation}|\bar{\mathbf{g}}^T\mathbf{f}\mathbf{w}_k|=\frac{|\bar{\mathbf{g}}^T\mathbf{f}|}{N|\mathbf{e}_n\mathbf{f}|},\label{gfw}
\end{equation}
where $|\mathbf{w}_k|=\frac{1}{N|\mathbf{e}_n\mathbf{f}|}$ for $\forall k$. 
Finally, by substituting \eqref{gfw} in \eqref{psi^t}, we obtain the optimal RIS phase shift vector with unit magnitude elements as
  \begin{align}
      {\boldsymbol{\psi}}_{\rm opt}^T=\frac{N|\mathbf{e}_n\mathbf{f}|}{|\mathbf{\bar{g}}^T\mathbf{f}|}\mathbf{\bar{g}}^T\mathbf{f}\mathbf{w}\label{OptSol_Phi_R1_IID}.
  \end{align}
\emph{2) Optimal Beamformer}: For a given phase shift matrix $\mathbf{\Phi}_{\rm opt}$, the optimization problem with respect to the beamforming vector $\mathbf{f}$ becomes
\begin{subequations}
\begin{align} 
    \max_{\mathbf{f}}~~ &|\kappa_l^2 \boldsymbol{\psi}_{\rm opt}^T{\rm diag}(\mathbf{\bar{h}}) \bar{\mathbf{H}} \mathbf{f} +\mu \kappa_l \bar{\mathbf{g}}^T\mathbf{f}|^2 + \kappa_l^2\kappa_n^2\|\bar{\mathbf{H}}\mathbf{f}\|^2,\label{subprob_f_R1_IID}\\
   \text{s.t.}~~&\|\mathbf{f}\|=1.
\end{align}   
\end{subequations}
To solve \eqref{subprob_f_R1_IID}, we start by substituting $\boldsymbol{\psi}_{\rm opt}$ in the first term of \eqref{subprob_f_R1_IID} as follows
 \begin{align}
     |\kappa_l^2 {\boldsymbol{\psi}_{\rm opt}}^T\mathbf{Ef} +\mu \kappa_l \mathbf{\bar{g}}^T\mathbf{f}|^2
     &=\kappa_l^4 |{\boldsymbol{\psi}_{\rm opt}}^T\mathbf{Ef}|^2 +\mu^2 \kappa_l^2 |\mathbf{\bar{g}}^T\mathbf{f}|^2+2\kappa_l^3\mu |{\boldsymbol{\psi}_{\rm opt}}^T\mathbf{Ef}||\mathbf{\bar{g}}^T\mathbf{f}|,\nonumber\\
     &\stackrel{(a)}{=}N^2\kappa_l^4|\mathbf{e}_n \mathbf{f}|^2 +\mu^2 \kappa_l^2 |\mathbf{\bar{g}}^T\mathbf{f}|^2+2N\kappa_l^3\mu |\mathbf{e}_n\mathbf{f}||\mathbf{\bar{g}}^T\mathbf{f}|,\label{term1}
 \end{align}
where step (a) follows from $\mathbf{wEf}=1$ and \eqref{psi^t}. 
The second term of \eqref{subprob_f_R1_IID} is simplified as
 \vspace{-0.2cm}
 \begin{align}     \|\bar{\mathbf{H}}\mathbf{f}\|^2&=\mathbf{f}^H\bar{\mathbf{H}}^H\bar{\mathbf{H}}\mathbf{f}=\mathbf{f}^H\bar{\mathbf{H}}^H{\rm diag}(\bar{\mathbf{h}})^H{\rm diag}(\bar{\mathbf{h}})\bar{\mathbf{H}}\mathbf{f},=\mathbf{f}^H\mathbf{E}^H\mathbf{Ef}=N|\mathbf{e}_n\mathbf{f}|^2,\label{term2}
 \end{align}
where  the second equality follows from  ${\rm diag}(\bar{\mathbf{h}})^H{\rm diag}(\bar{\mathbf{h}})=\mathbf{I}_{\rm N}$ and the last equality follows from \eqref{Ef}, respectively. Combining \eqref{term1} and \eqref{term2}, the objective given in \eqref{subprob_f_R1_IID} becomes
\begin{align}
w_1|\mathbf{e}_n\mathbf{f}|^2+w_2|\bar{\mathbf{g}}^T\mathbf{f}|^2+w_3|\mathbf{e}_n\mathbf{f}||\bar{\mathbf{g}}^T\mathbf{f}| = 
w_1\mathbf{f}^H\mathbf{E}_1\mathbf{f}+w_2\mathbf{f}^H\mathbf{Gf}+w_3|\mathbf{f}^H\mathbf{E}_2\mathbf{f}|,\label{modified_objfun}
\end{align} 
 where $\mathbf{E}_1=\mathbf{e}_n^H\mathbf{e}_n$, $\mathbf{G}=\bar{\mathbf{g}}^*\bar{\mathbf{g}}^T$, $\mathbf{E}_2=\mathbf{e}_n^H\bar{\mathbf{g}}^T$,  $w_1=N^2\kappa_l^4+N\kappa_l^2\kappa_n^2$, $w_2=\mu^2\kappa_l^2$, and $w_3=2N\mu\kappa_l^3$.\\
 It is to be noted that $\mathbf{E}_1$ and $\mathbf{G}$ are  symmetric and positive semidefinite matrices, whereas the matrix $\mathbf{E}_2$ is a negative definite matrix. Thus, the presence of the third term in \eqref{modified_objfun} makes the problem \eqref{subprob_f_R1_IID} non-convex. For this reason, we ignore the last term in \eqref{modified_objfun} from the maximization problem. This new objective will be equivalent to maximizing the lower bound on the mean {\rm SNR}.  It  is to be noted that this lower bound will be tight due to the following two reasons: 1) $w_1 \gg w_3$ and $w_2 \gg w_3$ and 2) the eigenvalues of  $\mathbf{E}_1$ and $\mathbf{G}$ are larger than the  $|\mathbf{f}^H\mathbf{E}_2\mathbf{f}|$. 
Using these arguments, we simplify the   problem for  maximizing the lower bound of mean {\rm SNR} as 
\vspace{-0.2cm}
\begin{align}
    \max_{\mathbf{f}} ~~& \mathbf{f}^H\mathbf{Zf},\\
   \text{s.t.} ~~& \|\mathbf{f}\|^2 = 1,
\end{align}
where  $\mathbf{Z}=w_1\mathbf{E}_1+w_2\mathbf{G}$ is a symmetric matrix. Thus, this optimization problem is equivalent to the Rayleigh quotient maximization, whose solution, i.e. the optimal beamformer,  becomes the dominant eigenvector of $\mathbf{Z}$ and can be given as 
\vspace{-0.3cm}
\begin{equation}
    \mathbf{f}_{\rm opt} = \mathbf{v_Z}.\label{OptSol_f_R1_iid}
\end{equation}
Substituting above $\mathbf{f}_{\rm opt}$ and \eqref{term2} in \eqref{OptSol_Phi_R1_IID} will further simplify the optimal phase shift vector $\boldsymbol{\psi}_{\rm opt}$.
We summarize the optimal beamforming for the independent fading in the following theorem.
\vspace{-0.3cm}
\begin{theorem}\label{Theo2}
    The SCSI-based optimal transmit beamformer and RIS phase shift matrix that maximizes the lower bound of mean SNR given in \eqref{OptProb_R1_IID} under {\rm i.i.d.} Rician-Rician fading are
    \begin{align}
    \mathbf{f}_{\rm opt}=\mathbf{v_Z}\text{~~~and~~~~} {\boldsymbol{\psi}}_{\rm opt}^T= \frac{\mathbf{\bar{g}}^T \mathbf{v_Z}}{|\mathbf{\bar{g}}^T \mathbf{v_Z}|}\frac{\mathbf{v_Z}^H\mathbf{E}^H}{|\mathbf{e}_n \mathbf{v_Z}|},\label{eq:optima_fPhi_R1IIDLB}
    \end{align}
    where $\mathbf{Z}=(N^2\kappa_l^4+N\kappa_l^2\kappa_n^2)\mathbf{e}_n^H\mathbf{e}_n + 2N\mu\kappa_l^3 \bar{\mathbf{g}}^*\bar{\mathbf{g}}^T$.
\end{theorem}
Now, we perform outage and capacity analysis for the optimal beamforming scheme given in Theorem \ref{Theo2} for this case of fading scenario in the following corollaries.
\vspace{-0.3cm}
\begin{corollary}
\label{cor:outage_R1IID}
    OP of the SCSI-based optimal beamforming scheme for the RIS-aided MISO system under {\rm i.i.d.} Rician-Rician fading is given by
    \begin{align}
        {\rm P_{out}}(\beta) \approx 1-Q_1\left(\frac{|m|}{\sqrt{\sigma/2}},\frac{\sqrt{{\beta}/{\gamma}}}{\sqrt{\sigma/2}}\right),\label{Pout_R1_IID}
    \end{align}
   where $m=N\kappa_l^2\mathbf{e}_n\mathbf{v_Z}+ \mu\kappa_l\bar{\mathbf{g}}^T\mathbf{v_Z}$,  $\sigma^2= N\kappa_n^2(1+\kappa_l^2|\mathbf{e}_n\mathbf{v_Z}|^2)+\mu^2\kappa_n^2$,
    and $\mathbf{Z}$ is given in  \eqref{eq:optima_fPhi_R1IIDLB}.
\end{corollary}
\begin{proof}
From \eqref{SNR_dist}, we have $ |\xi_1 + \mu\xi_2|\sim{\rm Rice}\left(|m|,\sigma\right)$ whose parameters given in \eqref{parameters_m_sigma} becomes  $m=\kappa_l^2\mathbf{\Bar{h}}^T\mathbf{\Phi \bar{H}f}+ \kappa_l\mathbf{\Bar{g}}^T\mathbf{f}$ and $\sigma^2 = \kappa_l^2\kappa_n^2\|\mathbf{\Bar{H}f}\|^2 +  (\kappa_l^2\kappa_n^2+\kappa_n^4)N + \mu^2\kappa_n^2$ for {\rm i.i.d}. fading scenario. Further, substituting $\mathbf{f}_{\rm opt}$ and $\boldsymbol{\psi}_{\rm opt}$ from \eqref{eq:optima_fPhi_R1IIDLB} and simplifying, completes the proof.
\end{proof}
\vspace{-0.3cm}
\begin{corollary}
    EC of the SCSI-based beamforming scheme proposed in  Theorem \ref{Theo2} can be determined approximately using \eqref{ergodic capacity} with ${\rm P_{out}}(\beta)$ given in Corollary \ref{cor:outage_R1IID}.
\end{corollary}
\vspace{-0.3cm}
\vspace{-0.3cm}
\subsection{Correlated Rician-Rayleigh fading}
\label{R2}
In this subsection, we assume the LoS component along the direct link to be absent, i.e. $\mathbf{g} = \mathbf{\Tilde{R}}_{\rm BT} \mathbf{\Tilde{g}}_{\rm W}$, reducing the considered model to correlated Rician-Rayleigh fading. This fading scenario with multiple receiving antennas was considered in \cite{WangJinghe_CorrelatedFading_2021} where the authors  proposed an SDR-based iterative algorithm for SCSI-based optimal beamforming. We will present a similar scheme along with its performance analysis. 

The mean SNR given in \eqref{obj_fun} for this fading scenario is given by
\begin{align}
    \Gamma(\mathbf{f}, \mathbf{\Phi}) = |\kappa_l^2\mathbf{\Bar{h}}^T\mathbf{\Phi\Bar{H}f}|^2 + \kappa_l^2\kappa_n^2\boldsymbol{\psi}^H\mathbf{Z}_1\boldsymbol{\psi} +\mathbf{f}^H\mathbf{R}_{\rm BT}\mathbf{f}[\mu^2+\kappa_n^2\boldsymbol{\psi}^H\mathbf{Z}_2\boldsymbol{\psi}].\label{OptProb_R2}
\end{align}
The above expression follows from steps given in Appendix \ref{AppA} with $\mathbf{g}=\mathbf{\tilde{g}}$.
 As $\mathbf{f}$ and $\mathbf{\Phi}$ are coupled, we tackle this scenario by dividing the problem into optimal beamformer and phase shift matrix sub-problems as below.\newline
 \emph{1) Optimal Beamformer}: For a given $\mathbf{\Phi}$, the optimization problem w.r.t $\mathbf{f}$ becomes
 \vspace{-0.3cm}
    \begin{subequations}
    \begin{align}
    \max_{\mathbf{f}} ~~& \mathbf{f}^H\mathbf{F}_{\rm s}\mathbf{f}, \label{subprob_f_R2}\\
   \text{s.t.} ~~& \|\mathbf{f}\|^2 = 1,
    \end{align}
    \end{subequations}
where the objective function follows from \eqref{OptProb_R2} with $\mathbf{F}_{\rm s} = \mathbf{F}_{\rm 1s} + \mathbf{F}_{\rm 2s} + \mathbf{F}_{\rm 3s}$
such that $\mathbf{F}_{\rm 1s} = \kappa_l^4\mathbf{E}^H\boldsymbol{\psi}^*\boldsymbol{\psi}^T\mathbf{E}$, $\mathbf{F}_{\rm 2s} = \kappa_l^2\kappa_n^2\mathbf{\Bar{H}}^H\mathbf{\Phi}^H\mathbf{R}_{\rm RT}\mathbf{\Phi\Bar{H}}$, $\mathbf{F}_{\rm 3s} = \mathbf{R}_{\rm BT}[\mu^2+\kappa_n^2\boldsymbol{\psi}^H\mathbf{Z}_2\boldsymbol{\psi}]$.
Note that $\mathbf{F}_{\rm 1s}$, $\mathbf{F}_{\rm 2s}$, and $\mathbf{F}_{\rm 3s}$ directly follow from $\mathbf{F}_1$, $\mathbf{F}_2$, and $\mathbf{F}_3$ given in Section \ref{optimal beamforming} by setting $\mathbf{\Bar{g}} = 0$. Thus, the symmetricity of $\mathbf{F}_{\rm s}$ also follows $\mathbf{F}$. Hence, we have
\vspace{-0.3cm}
\begin{equation}
    \mathbf{f}_{\rm opt} = \mathbf{v}_{\mathbf{F}_{\rm s}}.\label{OptSol_f_R2}
\end{equation}
\emph{2) Optimal Phase Shift Matrix}: For a given $\mathbf{f}$, the optimization problem w.r.t  $\mathbf{\Phi}$ becomes
\begin{subequations}
    \begin{align}
    \max_{\boldsymbol{\psi}}  ~~& \boldsymbol{\psi}^H \mathbf{A}_{\rm s} \boldsymbol{\psi}, \label{subprob_phi_R2}\\
   \text{s.t.} ~~& |\mathbf{\psi}_k| = 1 ~~ \forall k = 0,\ldots,N-1,
    \end{align}
\end{subequations}
where $\mathbf{A}_{\rm s}=\kappa_l^4 \mathbf{E}^*\mathbf{f}^* \mathbf{f}^T\mathbf{E}^T + \kappa_l^2 \kappa_n^2 \mathbf{Z}_1 + \kappa_n^2 \mathbf{f}^H\mathbf{R}_{\rm BT} \mathbf{f}~\mathbf{Z}_2$.  
Next, we reformulate the above problem using SDR as given in \eqref{relaxed_phi} with $\mathbf{A}=\mathbf{A}_{\rm s}$ and $\mathbf{\Psi}=\boldsymbol{\psi\psi}^H$ and obtain optimal  $\boldsymbol{\psi}$ for a given $\mathbf{f}$.  

As the  above solutions of $\mathbf{f}$ and $\mathbf{\Phi}$ are coupled, the optimal beamformer for this fading case can be obtained using Algorithm \ref{Alg1} by simply setting  $\mathbf{F}=\mathbf{F}_{\rm s}$, $\mathbf{A}=\mathbf{A}_{\rm s}$, and $\mathbf{\Psi}=\boldsymbol{\psi\psi}^H$. Further, OP and EC of optimal beamforming under this fading scenario can be evaluated using Theorem \ref{Theo1} and Corollary \ref{cor:R1_EC}, respectively, with modified parameters 
\begin{align}
    m=\kappa_l^2\mathbf{\Bar{h}}^T \mathbf{\Phi}\Bar{\mathbf{H}}\mathbf{f}_{\rm opt} \text{~~and~~} \sigma^2=\kappa_l^2 \kappa_n^2 \boldsymbol{\psi}_{\rm opt} ^H \mathbf{Z}_1 \boldsymbol{\psi}_{\rm opt}  +  ( \mu^2 + \kappa_n^2 \boldsymbol{\psi}_{\rm opt} ^H \mathbf{Z}_2 \boldsymbol{\psi}_{\rm opt} )\mathbf{f}^H_{\rm opt}  \mathbf{R_{\rm BT}} \mathbf{f}_{\rm opt} .
\end{align}
The above parameters can be obtained by modifying the parameters of the distribution of  $\xi_2$ given in \eqref{eq:xi1xi2} with $\mathbf{g}=\mathbf{\tilde{g}}$ and further using it to get mean and variance of $\xi_1+\mu\xi_2$. $\mathbf{f}_{\rm opt}$ and $\boldsymbol{\psi}_{\rm opt}$ are obtained from \autoref{Alg1} with modified parameters as mentioned above. 

\vspace{-.4cm}
\subsection{IID Rician-Rayleigh Fading}\label{R2_IID}
In this subsection, we consider {\rm i.i.d.} Rician-Rayleigh fading along the indirect-direct links such $\mathbf{g}=\mathbf{\tilde{g}}\sim\mathcal{C}\mathcal{N}(0,\mathbf{I_M})$, $\mathbf{\Tilde{h}}\sim\mathcal{C}\mathcal{N}(0,\mathbf{I_N})$ and $\mathbf{\Tilde{H}}_{:,i}\sim\mathcal{C}\mathcal{N}(0,\mathbf{I_N})$. For this scenario, the mean SNR, i.e. the objective function \eqref{objective}, can be obtained using the steps given in Appendix \ref{AppA} as 
\begin{equation}
    \Gamma(\mathbf{f}, \mathbf{\Phi}) = |\kappa_l^2\mathbf{\Bar{h}}^T\mathbf{\Phi\Bar{H}f}|^2 + \kappa_l^2\kappa_n^2\|\mathbf{\Bar{H}f}\|^2 + N\kappa_l^2\kappa_n^2 + N\kappa_n^4 + \mu^2. 
\end{equation}
By substituting $\mathbf{\Bar{H}} = \mathbf{a}_N(\theta_{\rm ra})\mathbf{a}^T_M(\theta_{\rm bd}^{\rm i})$ and further simplifying, we can write the mean {\rm SNR} as
\begin{align}
    \Gamma(\mathbf{f}, \mathbf{\Phi}) = |\mathbf{a}^H_M(\mathbf{\theta}_{\rm bd}^{\rm i})\mathbf{f}|^2\left[|\kappa_l^2\boldsymbol{\psi}^T{\rm diag}(\mathbf{\bar{h}})\mathbf{a}_N(\mathbf{\theta}_{\rm ra})|^2+\kappa_l^2\kappa_n^2N\right]+ N\kappa_l^2\kappa_n^2 + N\kappa_n^4 + \mu^2.\label{OptProb_R2_IID}
\end{align}
It is clear from \eqref{OptProb_R2_IID} that $\mathbf{f}$ and $\mathbf{\Phi}$ are decoupled and the problem can be equivalently transformed into two independent sub-problems with respect to $\mathbf{f}$ and $\mathbf{\Phi}$. By simple applications of co-phasing and projections, the authors in \cite{hu2020statistical} obtain the optimal beamforming solutions for this fading scenario. But, for completeness, we reconstruct the results in the following theorem.
\vspace{-0.3cm}
\begin{theorem}\label{Theo4}
    The SCSI-based optimal transmit beamformer and RIS phase shift matrix that maximize the mean SNR given in \eqref{OptProb_R2_IID} under {\rm i.i.d.} Rician-Rayleigh fading are
   \begin{align}
    \mathbf{f}_{\rm opt} = \frac{1}{\sqrt{M}}\mathbf{a}_M^H(\mathbf{\theta}_{\rm bd}^{i}),
    \text{~~and~~}\boldsymbol{\psi_{\rm opt}} =  e^{-j\angle{{\rm diag}(\mathbf{\Bar{h}})\mathbf{a}_N(\mathbf{\theta}_{\rm ra})}}.\label{OptSol_R2_IID}
    \end{align}
\end{theorem}
\begin{proof}
    Since \eqref{OptProb_R2_IID} is decoupled in $\mathbf{f}$ and $\mathbf{\Phi}$, we can select $\mathbf{f}$ that maximizes $|\mathbf{a}^H_M(\mathbf{\theta}_{\rm bd}^{\rm i})\mathbf{f}|^2$ and $\mathbf{\Phi}$ that maximizes $|\kappa_l^2\boldsymbol{\psi}^T{\rm diag}(\mathbf{\bar{h}})\mathbf{a}_N(\mathbf{\theta}_{\rm ra})|^2$. For this, one can clearly see that the optimal solutions of $\mathbf{f}$ and $\mathbf{\Phi}$ would be the ones given in \eqref{OptSol_R2_IID}.
\end{proof}
Now, we present OP and EC that are achievable through the beamforming scheme given in Theorem \ref{Theo4} in the following corollaries.
\vspace{-0.3cm}
\begin{corollary}
\label{cor:outage_R2IID}
    OP of the SCSI-based optimal beamforming scheme for the RIS-aided MISO system under {\rm i.i.d.} Rician-Rayleigh fading is given by
    \begin{align}
        {\rm P_{out}}(\beta) \approx 1-Q_1\left(\frac{|m|}{\sqrt{\sigma/2}},\frac{\sqrt{\beta/\gamma}}{\sqrt{\sigma/2}}\right),\label{Pout_R2_IID}
    \end{align}
    where $m = \kappa_l^2N\sqrt{M}$ and
    $\sigma^2 = (M + 1) N\kappa_l^2 \kappa_n^2 + N \kappa_n^4 + \mu^2$.
\end{corollary}
\begin{proof}
   For $\mathbf{g} = \Tilde{\mathbf{g}}$, the parameters of OP given in \eqref{parameters_m_sigma} becomes  $m=\kappa_l^2\mathbf{\Bar{h}}^T\mathbf{\Phi \bar{H}f}$ and $\sigma^2=\kappa_l^2\kappa_n^2\|\mathbf{\Bar{H}f}\|^2 + N\kappa_l^2\kappa_n^2 + N\kappa_n^4 + \mu^2$. Further, substituting $\mathbf{f}_{\rm opt}$ and $\boldsymbol{\psi}_{\rm opt}$ from \eqref{OptSol_R2_IID} and using $\|\mathbf{\Bar{H}f_{\rm opt}}\|^2 = MN$, we obtain $m$ and $\sigma^2$ as given in \eqref{Pout_R2_IID}.
\end{proof}
\vspace{-0.3cm}
\begin{corollary}
    EC of the SCSI-based  beamforming scheme proposed in  Theorem \ref{Theo4} can be determined approximately using \eqref{ergodic capacity} with  ${\rm P_{out}}(\beta)$ given in Corollary \ref{cor:outage_R2IID}.
\end{corollary}
\vspace{-0.7cm}
\subsection{Correlated Rayleigh-Rayleigh Fading}\label{R3}
\vspace{-0.1cm}
In this subsection, we assume that the LoS components along both the direct and indirect links are absent. This reduces the channel model presented in Section \ref{channel model} to a correlated Rayleigh-Rayleigh scenario wherein $\mathbf{H} = \mathbf{\Tilde{H}}$, $\mathbf{h} = \mathbf{\Tilde{h}}$, and $\mathbf{g} = \mathbf{\Tilde{g}}$. It may be noted that this scenario is a special case of the generalized fading model given in Section \ref{channel model}  as it is discussed  earlier to be a limiting case of fading factor, i.e. $K\to 0$. For this scenario, the mean {\rm SNR} becomes
\begin{equation}
    \Gamma(\mathbf{f}, \mathbf{\Phi}) = \mathbf{f}^H\mathbf{R}_{\rm BT}\mathbf{f}[\boldsymbol{\psi}^H(\mathbf{R}_{\rm RR}\odot\mathbf{R}_{\rm RT})\boldsymbol{\psi}+\mu^2].\label{OPtProb_R3}
\end{equation}
The above equation directly follows from the steps given in Appendix \ref{AppA} by setting $\kappa_l = 0$ and $\kappa_n = 1$.
It can be  seen from \eqref{OPtProb_R3} that the terms pertaining to $\mathbf{f}$ and $\mathbf{\Phi}$ are decoupled. 
This is expected as, in the absence of LoS components, the optimal choice of transmit beamforming vector $\mathbf{f}$ will depend on fading covariance matrix associated with BS and the optimal choice of phase shift matrix $\mathbf{\Phi}$ will depend on  fading  covariance matrices associated with RIS. 
Therefore, the optimal choice of $\mathbf{f}$ and $\mathbf{\Phi}$ can be selected independently of each other. 

As the covariance matrix is symmetric, we can set $\mathbf{f}$ equal to $\mathbf{v_{R_{\rm BT}}}$ for maximizing the term $\mathbf{f}^H\mathbf{R}_{\rm BT}\mathbf{f}$ with unit norm constraint \eqref{constraint_f}. For  $\mathbf{f}=\mathbf{v_{R_{\rm BT}}}$, the  maximum value of this term is equal to  $\lambda_{\mathbf{R}_{\rm BT}}$ which is nothing but the maximum eigenvalue value of $\mathbf{R}_{\rm BT}$.
Now, we will maximize the other term corresponding to the phase shift matrix as below
\begin{subequations}
    \begin{align}
    \max_{\boldsymbol{\psi}}  ~~& \boldsymbol{\psi}^H(\mathbf{R}_{\rm RR} \odot \mathbf{R}_{\rm RT})\boldsymbol{\psi}, \label{subprob_phi_R3}\\
   \text{s.t.} ~~& |\mathbf{\psi}_k| = 1 ~~ \forall k = 0,\ldots,N-1,
    \end{align}\label{OptProb_Phi_R3}
\end{subequations}
\vspace{-0.1cm}
The unit modulus constraint makes it difficult to solve the problem directly, as mentioned earlier. However, we could obtain the optimal $\mathbf{\Phi}$ using the fact that the matrix $\mathbf{R}_{\rm RR}\odot\mathbf{R}_{\rm RT}$ in the objective function is a real. 
The objective function \eqref{subprob_phi_R3} can be rewritten as
    \begin{align*}
        \boldsymbol{\psi}^H(\mathbf{R}_{\rm RR} \odot \mathbf{R}_{\rm RT})\boldsymbol{\psi}=& \sum_{i,j}\mathbf{R}_{{\rm RT},{ij}}\mathbf{R}_{{\rm RR},{ij}}\boldsymbol{\psi}_i^H\boldsymbol{\psi}_j
        = {\rm trace}(\mathbf{R}_{\rm RR}\odot\mathbf{R}_{\rm RT}) + \sum_{i\neq j}\mathbf{R}_{{\rm RT},{ij}}\mathbf{R}_{{\rm RR},{ij}}\boldsymbol{\psi}_i^H\boldsymbol{\psi}_j.\label{phi_int} 
    \end{align*}
Note ${\rm trace}(\mathbf{R}_{\rm RR}\odot\mathbf{R}_{\rm RT})$ is a real scalar quantity and is independent of $\mathbf{\Phi}$, whereas the second term is the summation of complex scalars. Thus, we need to co-phase the complex numbers to maximize the second term. To do this, we can simply set $\boldsymbol{\psi}_i = \boldsymbol{\psi}_j$, $\forall i,j=1,\dots,N$.
Hence, the optimal $\mathbf{\Phi}$ can be obtained as $\boldsymbol{\psi}_{{\rm opt},k} = e^{j{\theta}}$ for $k=1,\dots,N$ where $\theta\in[-\pi/2 , \pi/2]$. For this choice of $\boldsymbol{\psi}_{\rm opt}$, maximum value of the objective given in \eqref{OptProb_Phi_R3} becomes
\begin{equation}
    {\rm trace}(\mathbf{R}_{\rm RR}\odot\mathbf{R}_{\rm RT}) + \sum\nolimits_{i\neq j}\mathbf{R}_{\rm RT_{ij}}\mathbf{R}_{\rm RR_{ij}} = \sum\nolimits_{i,j} \mathbf{R}_{\rm RT_{ij}}\mathbf{R}_{\rm RR_{ij}}=\mathbf{1}_{\rm N}^T(\mathbf{R}_{\rm RR}\odot\mathbf{R}_{\rm RT})\mathbf{1}_{\rm N}.\label{Max_Obj_R3}
\end{equation}
For $\mathbf{R}_{\rm RT} = \mathbf{R}_{\rm RR}=\mathbf{R}$, \eqref{Max_Obj_R3} becomes $\|\mathbf{R}\|_{\rm F}^2$. The above results are summarized in  Theorem \ref{Theo5}.
\vspace{-0.3cm}
\begin{theorem}\label{Theo5}
    The SCSI-based optimal transmit beamformer and RIS phase shift matrix that maximize the mean SNR given in \eqref{OptProb_R2_IID} under {\rm i.i.d.} Rayleigh-Rayleigh fading are
   \begin{align}
    \mathbf{f}_{\rm opt} = \mathbf{v_{R_{\rm BT}}}
    \text{~~and~~}\boldsymbol{\psi}_{\rm opt} = \mathbf{1}_{\rm N} e^{j{\theta}},\label{OptSol_R3}
    \end{align}
    \vspace{-0.3cm}
    where $\theta\in[-\pi/2 , \pi/2]$ and the maximum mean {\rm SNR} value is
    \begin{equation}
        \Gamma(\mathbf{f}_{\rm opt},\mathbf{\Phi}_{\rm opt})=\begin{cases}
            \lambda_{\mathbf{R}_{\rm BT}}\left(\mu^2 + \|\mathbf{R}\|_F^2\right), ~~&\text{if~ }\mathbf{R}_{\rm RT} = \mathbf{R}_{\rm RR}=\mathbf{R}\\
             \lambda_{\mathbf{R}_{\rm BT}}\left(\mu^2 +\mathbf{1}_{\rm N}^T(\mathbf{R}_{\rm RR}\odot\mathbf{R}_{\rm RT})\mathbf{1}_{\rm N}\right), ~~&\text{otherwise},
        \end{cases}\label{eq:R3cor_maxSNR}
    \end{equation}
    and $\lambda_{\mathbf{R}_{\rm BT}}$ is the maximum eigen value of $\mathbf{R}_{\rm BT}$.
\end{theorem}
Now, we present OP and EC achievable through the scheme given in Theorem \ref{Theo5} in the following corollaries.
\vspace{-0.3cm}
\begin{corollary}\label{cor:outage_R3cor}
    OP of the SCSI-based optimal beamforming scheme for the RIS-aided MISO system under correlated Rayleigh-Rayleigh fading is given by
    \begin{align}
        {\rm P_{out}}(\beta) \approx 1-Q_1\left(0,\frac{\sqrt{\beta/\gamma}}{\sqrt{\sigma/2}}\right),\label{Pout_R3}
    \end{align}
    where  $ m =0$ and $\sigma^2 =  \Gamma(\mathbf{f}_{\rm opt},\mathbf{\Phi}_{\rm opt})$, and  $\Gamma(\mathbf{f}_{\rm opt},\mathbf{\Phi}_{\rm opt})$ is given in \eqref{eq:R3cor_maxSNR}.
\end{corollary}
\begin{proof}
    For $\mathbf{H} = \mathbf{\Tilde{H}}$, $\mathbf{h} = \mathbf{\Tilde{h}}$, and $\mathbf{g} = \mathbf{\Tilde{g}}$, the parameters given in \eqref{parameters_m_sigma} becomes $m = 0$ and $\sigma^2 = \mathbf{f}^H\mathbf{R}_{\rm BT}\mathbf{f}[\boldsymbol{\psi}^H(\mathbf{R}_{\rm RR}\odot\mathbf{R}_{\rm RT})\boldsymbol{\psi}+\mu^2]$. Further, by substituting $\mathbf{f}_{\rm opt}$ and $\mathbf{\boldsymbol{\psi}_{\rm opt}}$ from \eqref{OptSol_R3} and simplifying, we obtain \eqref{Pout_R3}.
\end{proof}
\vspace{-0.3cm}
\begin{corollary}
    EC of the SCSI-based  beamforming scheme proposed in  Theorem \ref{Theo5} can be determined approximately using \eqref{ergodic capacity} with ${\rm P_{out}}(\beta)$ given in Corollary \ref{cor:outage_R3cor}.
\end{corollary}
\vspace{-0.6cm}
\subsection{IID Rayleigh-Rayleigh Fading}\label{R3_IID}
Here, we assume that  both the direct and indirect links undergo {\rm i.i.d} multipath fading with the absence of LoS components. This results in the {\rm i.i.d} Rayleigh-Rayleigh fading scenario such that 
$\mathbf{g}\sim\mathcal{C}\mathcal{N}(0,\mathbf{I_M})$, $\mathbf{h}\sim\mathcal{C}\mathcal{N}(0,\mathbf{I_N})$ and $\mathbf{H}_{:,i}\sim\mathcal{C}\mathcal{N}(0,\mathbf{I_N})$. For this, the mean {\rm SNR} reduces to
\vspace{-0.3cm}
\begin{equation}
  \Gamma(\mathbf{f},\mathbf{\Phi})=  \mu^2 + N\label{OptProb_R3_IID},
\end{equation}
which is independent of  $\mathbf{f}$ and $\mathbf{\Phi}$. This implies that the mean {\rm SNR}  is a constant value regardless of the choice of beamforming vector. Therefore, we can set $\mathbf{f}_{\rm opt} \in \mathcal{S}_\mathbf{f} = \{ \mathbf{f} \in \mathbb{C}^M : \|\mathbf{f}\| = 1 \}$ and $\boldsymbol{\psi}_{\rm opt} \in \mathcal{S}_{\boldsymbol{\psi}}  = \{ \boldsymbol{\psi} \in \mathbb{C}^N : |\boldsymbol{\psi}_k| = 1 ; \forall i = 1 \cdots N\}$.
Now, we present the outage performance of this proposed beamforming scheme in  the following corollary. 
\setcounter{theorem}{5}
\setcounter{corollary}{0}
\vspace{-0.3cm}
\begin{corollary}\label{cor:outage_R3IID}
For any $\mathbf{f}\in\mathcal{S}_\mathbf{f}$ and $\boldsymbol{\psi}\in\mathcal{S}_\mathbf{\boldsymbol{\psi}}$, OP for the RIS-aided MISO system under {\rm i.i.d.}. Rayleigh-Rayleigh fading is given by
\begin{align}
    {\rm P_{out}}(\beta) \approx 1-Q_1\left(0,\frac{\sqrt{\beta/\gamma}}{\sqrt{\sigma/2}}\right), \text{~where~}  m =  0 \text{~and~} \sigma^2=\mu^2 + N.\label{Pout_R3_IID}
\end{align}
\end{corollary}
\vspace{-0.3cm}
\begin{proof}
For $\mathbf{{H}} = \mathbf{\Tilde{H}}_{\rm W}$, $\mathbf{h} = \mathbf{\Tilde{h}}_{\rm W}, \mathbf{g} = \mathbf{\Tilde{g}}_{\rm W}$, the parameters given in \eqref{parameters_m_sigma} become $m = 0$ and $\|\mathbf{f}\|^2\left(\mu^2 + \|\boldsymbol{\psi}\|^2\right)$. Further, by substituting $\mathbf{f}_{\rm opt}$ and $\mathbf{\boldsymbol{\psi}_{\rm opt}}$ and simplifying, we obtain \eqref{OptProb_R3_IID}. 
\end{proof}
\vspace{-0.3cm}
\begin{corollary}
    For any $\mathbf{f}\in\mathcal{S}_\mathbf{f}$ and $\boldsymbol{\psi}\in\mathcal{S}_\mathbf{\boldsymbol{\psi}}$,  EC  under {\rm i.i.d.} Rayleigh-Rayleigh fading can be determined approximately using \eqref{ergodic capacity} with ${\rm P_{out}}(\beta)$ given in Corollary \ref{cor:outage_R3IID}.
\end{corollary}
\vspace{-0.3cm}
The optimal transmit beamforming vector $\mathbf{f}_{\rm opt}$, RIS phase shift matrix $\mathbf{\Phi}_{\rm opt}$, and the achievable OP along with its parameters $(m,\sigma^2)$ are summarized in Table \ref{OP_Table} for various fading cases studied in the above subsections. For easy referencing, we refer the Rician-Rician, Rician-Rayleigh, and Rayleigh-Rayleigh fading cases as R1, R2 and R3, respectively. The rows associated with fading cases that have algorithmic and closed-form beamforming solutions are highlighted in different colors.
It may be noted that the parameters $m$ and $\sigma^2$ are  useful to determine the mean SNR as $$\Gamma(\mathbf{f},\mathbf{\Phi})=\mathbb{E}[|\xi_1+\mu\xi_2|^2]=\sigma^2+|m|^2.$$
\begin{table}[h]
\vspace{-.4cm}
\caption{Summary of optimal beamforming and outages}
\label{OP_Table}
\vspace{-0.3cm}
\centering
\fontsize{19pt}{19pt}
\renewcommand{\arraystretch}{2.5}
\centering
\resizebox{\textwidth}{!}{
\begin{tabular}{|c|c|c|c|c|}

\hline
\multicolumn{1}{|c|}{} & \multicolumn{2}{|c|}{\textbf{Optimal Beamforming}} & \multicolumn{2}{|c|}{\textbf{Outage Probability} ${\rm P_{out}}(\beta)=1-Q_1\left(\frac{|m|}{\sqrt{\sigma/2}},\frac{\sqrt{\beta/\gamma}}{\sqrt{\sigma/2}}\right)$} \\
\cline{2-5}
\textbf{Fading Scenario}& {$\mathbf{f}_{\rm opt}$} &{$\mathbf{\Phi_{\rm opt}}$} & {$m$} & {$\sigma^2$}  \\

\hline
\multicolumn{1}{|c|}{\textbf{R1 Corr}} & \multicolumn{2}{|c|}{\Large{Algorithm \autoref{Alg1}}}&  $\kappa_l^2\mathbf{\Bar{h}}^T \mathbf{\Phi}_{\rm opt}\Bar{\mathbf{H}}\mathbf{f}_{\rm opt}  + \mu \kappa_l \Bar{\mathbf{g}}^T\mathbf{f}_{\rm opt}$ & $\kappa_l^2 \kappa_n^2 \boldsymbol{\psi}_{\rm opt}^H \mathbf{Z}_1 \boldsymbol{\psi}_{\rm opt} + \kappa_n^2 ( \mu^2 + \boldsymbol{\psi}_{\rm opt}^H  \mathbf{Z}_2 \boldsymbol{\psi}_{\rm opt})\mathbf{f}_{\rm opt}^H \mathbf{R_{\rm BT}} \mathbf{f}_{\rm opt}$\\

\hline
\multicolumn{1}{|c|}{\textbf{R1 IID} } & \multicolumn{2}{|c|}{\Large{Algorithm}  \ref{Alg2}}& $\kappa_l^2\mathbf{\Bar{h}}^T \mathbf{\Phi}_{\rm opt}\Bar{\mathbf{H}}\mathbf{f}_{\rm opt}  + \mu \kappa_l \Bar{\mathbf{g}}^T\mathbf{f}_{\rm opt}$ & $\kappa^2_l \kappa^2_n \|\Bar{\mathbf{H}}\mathbf{f}_{\rm opt}\|^2 + \kappa^2_n ( \mu^2 + N(\kappa^2_l  + \kappa^2_n))$\\ \hline

\rowcolor{lavender}
\textbf{R1 IID LB} & $\mathbf{v_Z}$ & $ \frac{\mathbf{\bar{g}}^T \mathbf{v_Z}}{|\mathbf{\bar{g}}^T \mathbf{v_Z}|}\frac{\mathbf{v_Z}^H\mathbf{E}^H}{|\mathbf{e}_n \mathbf{v_Z}|}$  & $N \kappa_l^2 \mathbf{e}_n \mathbf{v_Z} + \mu \kappa_l \Bar{\mathbf{g}}^T \mathbf{v_Z}$ & $N \kappa_n^2 (1 + \kappa_l^2 |\mathbf{e}_n \mathbf{v_Z}|^2) + \mu^2 \kappa_n^2$ \\ 

\hline
\multicolumn{1}{|c|}{\textbf{R2 Corr}} & \multicolumn{2}{|c|}{\Large{Algorithm \autoref{Alg1}} with \mbox{\normalsize $\mathbf{F}=\mathbf{F}_{\rm s}$, $\mathbf{A}=\mathbf{A}_{\rm s}$,  $\mathbf{\Psi}=\boldsymbol{\psi\psi}^H$} } & $\kappa_l^2\mathbf{\Bar{h}}^T \mathbf{\Phi}_{\rm opt}\Bar{\mathbf{H}}\mathbf{f}_{\rm opt}$ & $\kappa_l^2 \kappa_n^2 \boldsymbol{\psi}_{\rm opt}^H \mathbf{Z}_1 \boldsymbol{\psi}_{\rm opt} +  ( \mu^2 + \kappa_n^2 \boldsymbol{\psi}_{\rm opt}^H \mathbf{Z}_2 \boldsymbol{\psi}_{\rm opt})\mathbf{f}_{\rm opt}^H \mathbf{R_{\rm BT}} \mathbf{f}_{\rm opt}$ \\

\rowcolor{lavender}
\hline
\textbf{R2 IID} & $\frac{\mathbf{a}_M^H(\mathbf{\theta}_{\rm bd}^{i})}{\sqrt{M}}$ & $e^{-j\angle{{\rm diag}(\mathbf{\Bar{h}})\mathbf{a}_N(\mathbf{\theta}_{ra})}}$ & $N \sqrt{M} \kappa_l^2$ & $(M + 1) N\kappa_l^2 \kappa_n^2 + N \kappa_n^4 + \mu^2$ \\ \hline

\rowcolor{lavender}
\textbf{R3 Corr} & $\mathbf{v_{R_{\rm BT}}}$ & $\mathbf{1}_{\rm N} e^{j{\theta}}$ & $0$ & $\lambda_{\mathbf{R}_{\rm BT}} (\mu^2 + \|\mathbf{R}\|_F^2)$  \\ \hline

\rowcolor{lavender}
\textbf{R3 IID} & $\mathbf{f}_{\rm opt} \in \mathcal{S}_\mathbf{f}$ & $\boldsymbol{\psi}_{\rm opt} \in \mathcal{S}_{\boldsymbol{\psi}}$ & $0$ & $\mu^2 + N$  \\
\hline

\multicolumn{5}{|c|}{\textbf{*}~~$\mathbf{Z}_1 = \mathbf{R}_{\rm RT}\odot\mathbf{\Bar{H}f_{\rm opt}f_{\rm opt}}^H\mathbf{\Bar{H}}^H,~\mathbf{Z}_2 = \mathbf{R}_{\rm RR}\odot( \kappa_n^2\mathbf{R}_{\rm RT}+\kappa_l^2\mathbf{\Bar{h}}^*\mathbf{\Bar{h}}^T),~\lambda_{\mathbf{R}_{\rm BT}} = \lambda_{\rm max}\{\mathbf{R}_{\rm BT}\},$ $\mathbf{f}_{\rm opt}~\&~\boldsymbol{\psi}_{\rm opt}$ \Large{are obtained using the corresponding algorithms,}} \\

\multicolumn{5}{|c|}{\hspace{-13cm}$\mathbf{Z} ~~~=(N^2\kappa_l^4+N\kappa_l^2\kappa_n^2)\mathbf{e}_n^H\mathbf{e}_n + 2N\mu\kappa_l^3 \bar{\mathbf{g}}^*\bar{\mathbf{g}}^T$, \Large{and}~$\mathbf{e}_n=\mathbf{E}_{n,:}$ is the $n$-th row of $\mathbf{E} = {\rm diag}(\mathbf{\bar{h}})\mathbf{\Bar{H}}.$} \\
\hline
\end{tabular}}
\vspace{-0.5cm}
\end{table}
The following remarks present important insights derived using the  summary given in Table \ref{OP_Table} .
\begin{remark}\label{remark1}
From Table \ref{OP_Table}, it can be noted that the parameter $m=\mathbb{E}[\xi_1+\mu\xi_2]$ is equal to zero for Rayleigh-Rayleigh (R3) fading case as the coefficients of channels  $\mathbf{g}$, $\mathbf{h}$, and $\mathbf{H}$ are zero-mean complex Gaussian distributed. 
Therefore, the mean SNR is given by  
$\Gamma(\mathbf{f},\mathbf{\Phi})=\mathbf{E}[|\xi_1+\mu\xi_2|^2]=\sigma^2$. However, the  maximum mean SNRs achievable via SCSI-based optimal beamforming corresponding to  {\rm i.i.d.} and correlated scenarios of R3 fading case are $\mu^2+N$ and $\lambda_{\mathbf{R}_{\rm BT}}(\mu^2+\|\mathbf{R}\|_F^2)$. 
Thus, using $\lambda_{\mathbf{R}_{\rm BT}}>1$  and $\|\mathbf{R}\|_F^2>N$, it is safe to conclude that the maximum mean SNR under correlated scenario is higher than it is under {\rm i.i.d} scenario.  This is quite evident as the optimal beamforming scheme for correlated scenarios can leverage the information of covariance matrices to maximize the mean SNR.  
However, under {\rm i.i.d.} scenario, the mean SNR does not rely on the choice of beamforming vectors because of the  circularly symmetric fading. Therefore, we can say that the {\rm i.i.d.} and  fully correlated fadings are the extreme scenarios wherein the corresponding achievable mean SNR realizes its extreme values such that
$$\mu^2+N\leq \Gamma(\mathbf{f}_{\rm opt},\mathbf{\Phi}_{\rm opt})\leq M(\mu^2+N^2),$$
where the upper bound corresponds to the fully correlated scenario for which $\lambda_{\mathbf{R}_{\rm BT}}=M$ and $\|\mathbf{R}\|_F^2=N^2$. It is interesting to note  that the maximum mean SNR under correlated R3 fading increases with order  between linear and quadratic in  the number of RIS elements
$N$.  
\end{remark}
\vspace{-0.4cm}
\begin{remark}\label{remark2}
    Using the parameters $m$ and $\sigma^2$ given in Table \ref{OP_Table}, the maximum  mean SNRs achievable via SCSI-based optimal beamforming under {\rm i.i.d.} scenario can be determined as 
    \begin{align*}
       \Gamma(\mathbf{f}_{\rm opt},\mathbf{\Phi}_{\rm opt})= \begin{cases}
            | N \kappa_l^2 \mathbf{e}_n \mathbf{v_Z} + \mu \kappa_l \Bar{\mathbf{g}}^T \mathbf{v_Z}|^2 + N \kappa_n^2 (1 + \kappa_l^2 |\mathbf{e}_n \mathbf{v_Z}|^2) + \mu^2 \kappa_n^2, ~~&\text{for R1 (LB) case}\\
            N^2M\kappa_l^4+ N((M+1)\kappa_l^2\kappa_n^2 + \kappa_n^2)+\mu^2, ~~&\text{for R2 case}\\
            \mu^2+N, ~~&\text{for R3 case}
        \end{cases}%
    \end{align*}
    
The maximum mean SNRs under R1 and R2 cases increase quadratically with the number of RIS elements $N$, while  it increases linearly with $N$ under R3. This is because R1 and R2 cases include LoS components which allow efficient selection of beamformer for greater gains. In contrast, the absence of LoS components in R3 only increases the number of paths with $N$, resulting in linear improvement in mean SNR with $N$. In other words, the inclusion of deterministic LoS components in R1 and R2 increases the mean $m$ and variance $\sigma^2$  of $\xi_1 + \mu \xi_2$ compared to the R3 case, as can be verified from Table \ref{OP_Table}. This allows the beamformer to efficiently maximize $\Gamma(\mathbf{f},\mathbf{\Phi})=\mathbb{E}[|\xi_1+\mu\xi_2|^2]=\sigma^2+|m|^2$ further w.r.t $\mathbf{f}$ and $\mathbf{\Phi}$. \newline
In addition, it can be seen that the maximum mean SNR under the R1 case depends on the dot products of vectors $\mathbf{e}_n$ and $\mathbf{\bar{g}}$ with $\mathbf{v_Z}$, which are determined by the AoDs/AoAs of the LoS components in the direct and indirect links. 
Furthermore, as expected, the maximum mean SNR in R2 reduces to that of R3 as $K\to 0$, i.e. as the significance of LoS component in R2 diminishes.
       
\end{remark}
\vspace{-0.7cm}
\section{Numerical Results and Discussion}\vspace{-.1cm}
This section presents the performance of the proposed SCSI-based beamforming schemes for different fading scenarios and their comparisons with the PCSI-based SNR maximization scheme. 
First, we verify the derived outage performances of these schemes via simulations for both {\rm i.i.d.} and correlated scenarios. Then, we will discuss their achievable ECs for various system parameters. For the numerical analysis, we set the parameters as follows: the number of BS antennas $M = 4,$ number of RIS elements $N = 32$, Rician factor $K = 2$, the AoD from  BS to  the user via  direct link $\theta_{\rm{bd}}^{\rm{d}} = 0^\circ$, the AoD from  BS to  RIS $\theta_{\rm{ra}} = \pi/4$, the AoD from  RIS to  user $\theta_{\rm{rd}} = 8 \pi/5$, $\gamma = 1$, and PLR $\mu = 5~\rm{dB}$ unless mentioned otherwise. 
\begin{figure}[ht!]
\centering\vspace{-.5cm}
\begin{minipage}{.4\textwidth}
  \centering
  \includegraphics[width=\textwidth]{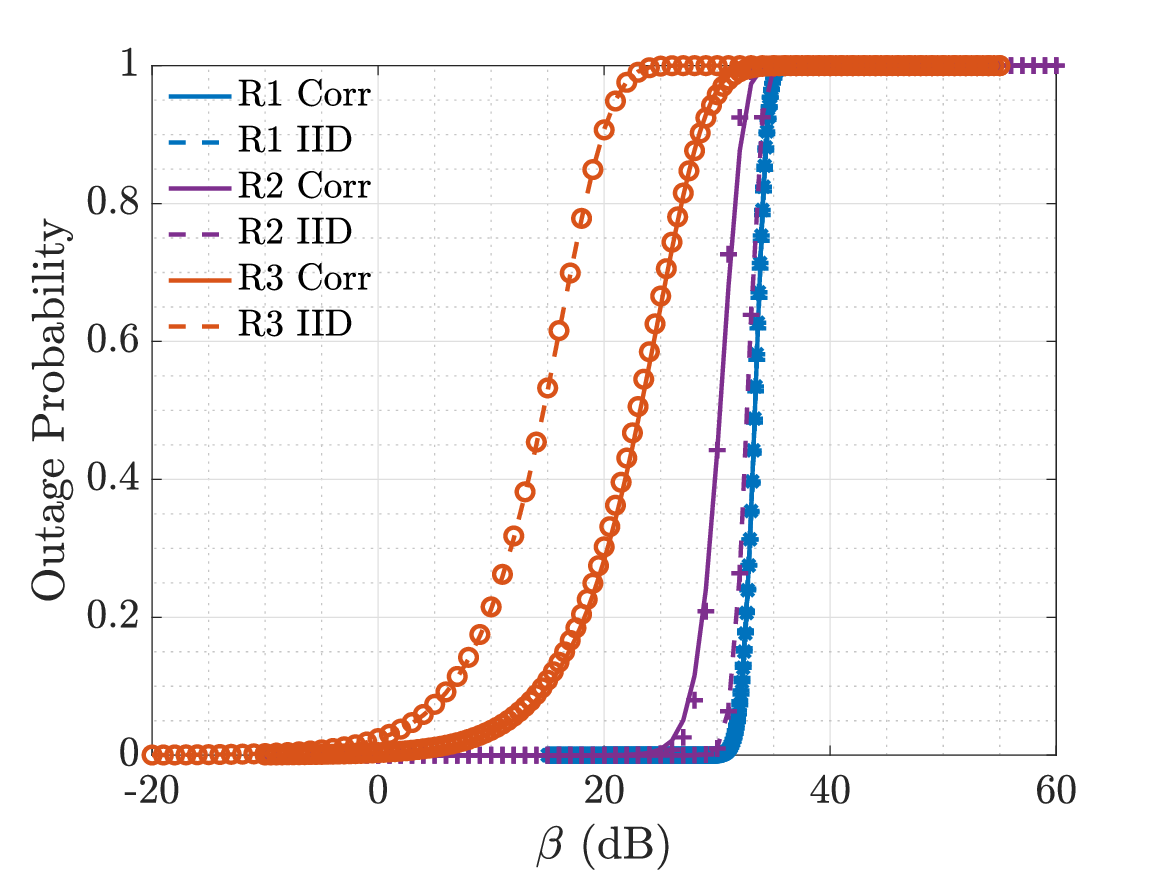}
\end{minipage}%
\begin{minipage}{.4\textwidth}
  \centering
  \includegraphics[width=\textwidth]{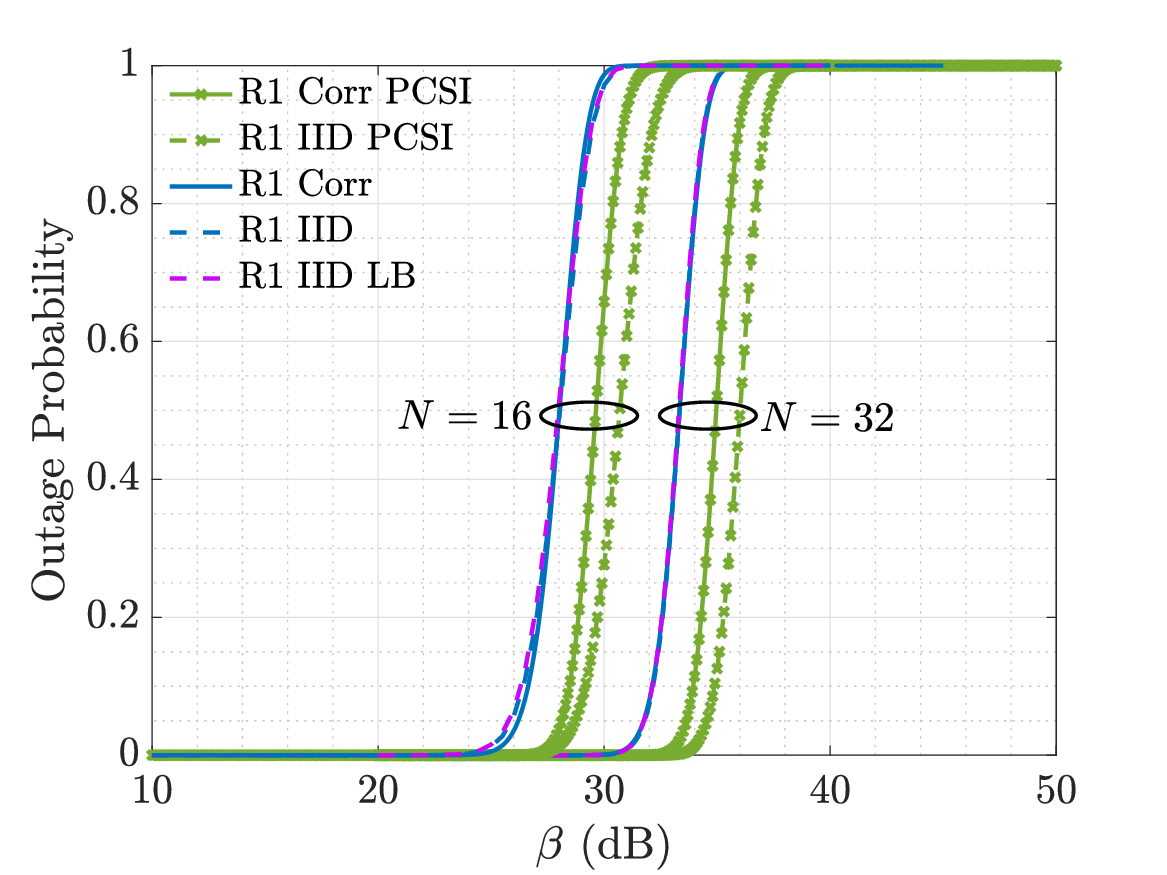} 
\end{minipage}\vspace{-.2cm}
\caption{Left: OP verification for R1, R2, R3 cases. Right: LB OP accuracy for the R1 case and its comparison with the PCSI scheme. Markers represent the simulation results, whereas the solid lines represent the analytical results.}\vspace{-.5cm}
\label{Fig1}
\end{figure}

Fig. \ref{Fig1} shows the outage performance of the proposed beamforming schemes for different fading scenarios. It may be noted that OPs for R1 {\rm i.i.d.} lower bound, R2 {\rm i.i.d.}, and R3 are derived in \eqref{Pout_R1_IID}, \eqref{Pout_R2_IID}, and \eqref{Pout_R3} respectively. However, OPs for R1 Corr and R2 Corr are derived for given transmit beamformer $\mathbf{f}$ and RIS phase shift matrix $\mathbf{\Phi}$ (refer to Tabel \ref{OP_Table}). Thus, we use the statistically optimal $\mathbf{f}$ and $\mathbf{\Phi}$ obtained through the proposed algorithms to evaluate the outage for the corresponding scenarios. Fig. \ref{Fig1} (Left) shows the accuracy of the derived approximate OP expressions under all fading scenarios where they closely match the simulation results. It can be observed that OP improves under R3, R2, and R1 cases in order. This improvement is because the statistically fixed beamformer performs better in the presence of strong LoS components, which is the case in R1 and R2. Furthermore, it can be seen that the correlated fading scenario provides better OP than {\rm i.i.d.} under R3, whereas the trend is  opposite under R2. This is because the statistical beamformer can exploit the correlated channel under R3, as is evident from \autoref{OP_Table}. In the presence of LoS under R2, the statistically fixed beamformer can utilize the independent fading along with the direct link for higher gains. However, OP is equal for both correlated and {\rm i.i.d.} scenarios under R1. Fig. \ref{Fig1} (Right) verifies the LB OP accuracy derived for the R1 {\rm i.i.d.} case and compares it with the PCSI scheme. The figure shows that the outage LB is tight for $N = 16$ and $N = 32$. As noted above, the SCSI-based scheme performs equally good for correlated and {\rm i.i.d.} scenarios under R1 fading. However, the PCSI scheme provides a better outage under the {\rm i.i.d.} scenario. This is because the instantaneous beamformer under the PCSI scheme with {\rm i.i.d.} scenario utilizes the large spectral range of the channel matrix. 
 \begin{figure}[ht!]
\centering\vspace{-.5cm}
\begin{minipage}{0.4\textwidth}
  \centering
  \includegraphics[width=\textwidth]{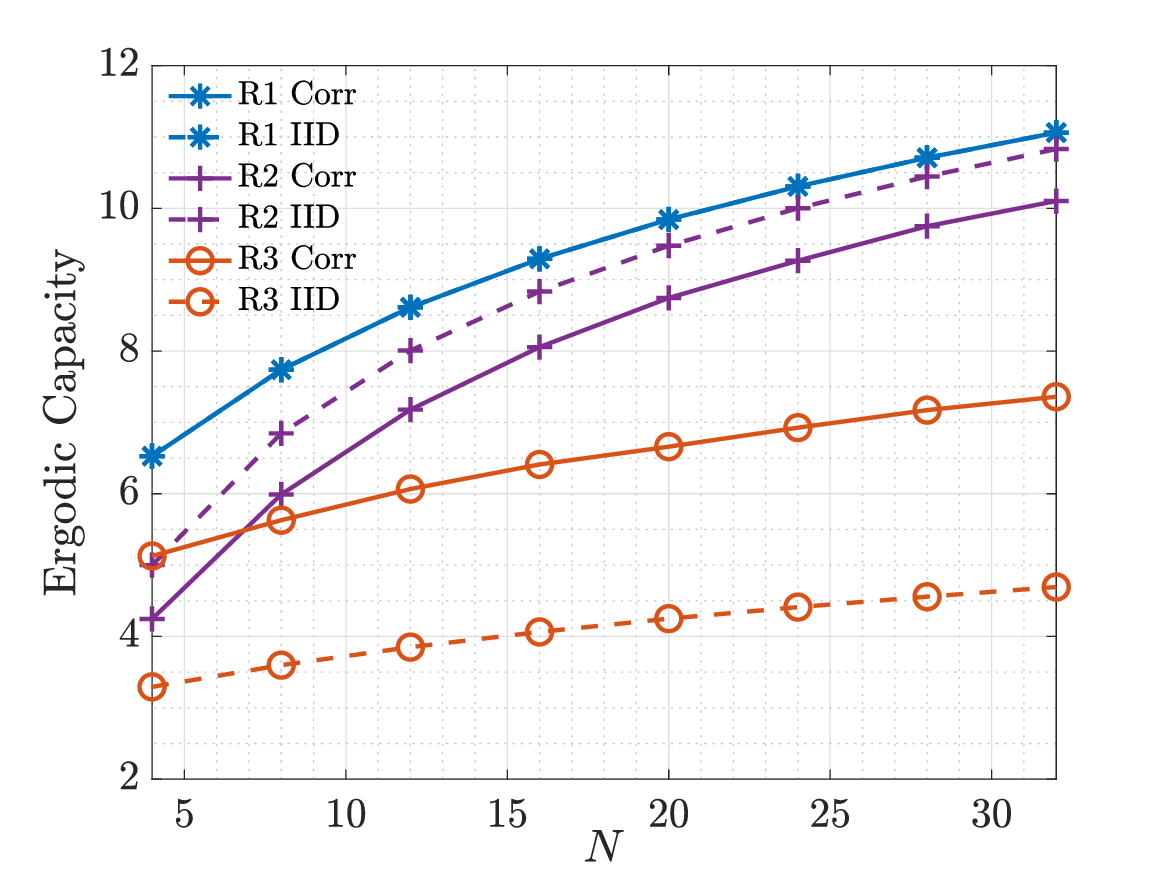}
\end{minipage}%
\begin{minipage}{0.4\textwidth}
  \centering
  \includegraphics[width=\textwidth]{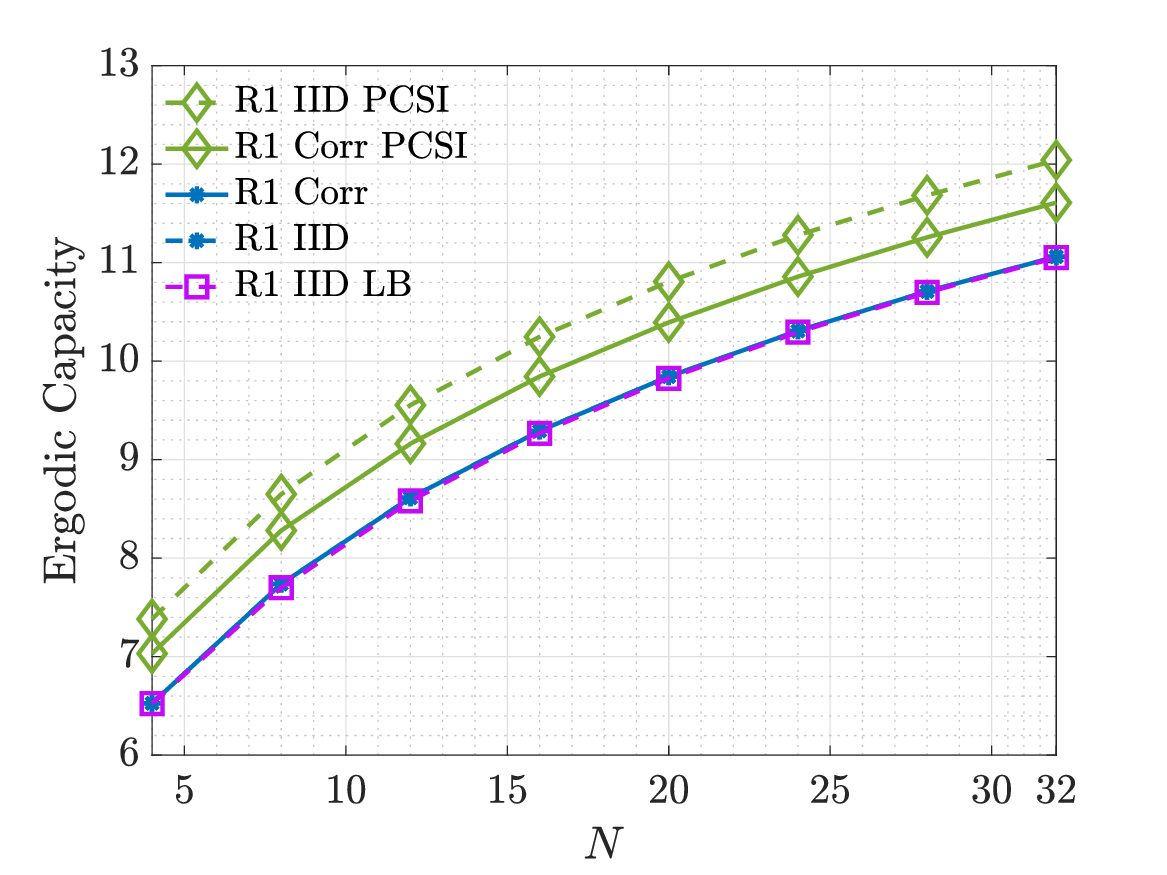}
\end{minipage}\vspace{-.2cm}
\caption{Left: EC vs. $N$ under R1, R2, R3 cases. Right: LB EC accuracy for R1 case and its comparison with the PCSI scheme.}\vspace{-.5cm}
\label{Fig2}
\end{figure}

Fig. \ref{Fig2} shows EC performance with respect to the number of RIS elements $N$ under R1, R2, and R3 cases. In Fig. \ref{Fig2} (Left), we see that EC increases with $N$ in all the cases, as expected. Particularly, the capacity increases rapidly in R1 and R2 cases compared to R3. 
This may be because the SCSI-based scheme  for R1 and R2 cases almost follows  performance trends similar to  the PCSI-based scheme for which the SNR is known to improve quadratically with $N$. Therefore, one can expect that the mean SNR will also improve quadratically in $N$ under R1 and R2 cases. In fact, this is shown to be the case in Remark \ref{remark2} for {\rm i.i.d.} scenario. However, the mean SNR in R3 improves with order between linear  and quadratic   in terms of $N$, as can be verified using Remarks \ref{remark1}.
Further, it can be observed that the performance difference between R1 and R2 saturates with the increase in $N$. This is because the large $N$ compensates for the loss due to the absence of LoS components along the direct link in R2. In other words, EC gain is contributed by two factors: 1) the number of reflecting elements $N$ and 2) the presence of LoS components. As both of these factors are present in the R1 case, it naturally outperforms the R2 and R3 cases. However, in R2, the LoS component along the direct link is absent, which reduces its capacity performance relatively. Nonetheless, the huge gain achieved by increasing $N$ helps to compensate for this relative loss. Therefore, the difference in gains under R1 and R2 saturates. Furthermore, it can be observed that the capacity performance trends in terms of the correlated and {\rm i.i.d.} scenarios for R1, R2, and R3 cases are similar to their outage performances as discussed in Fig. \ref{Fig1}. Fig. \ref{Fig2} (Right) verifies that the LB EC derived in \eqref{Pout_R1_IID} for the R1 {\rm i.i.d.} case is tight for a wide range of $N$. Moreover, it can be observed that the PCSI scheme provides better capacity than the SCSI scheme, as expected. However, unlike the correlated and {\rm i.i.d.} scenarios of R1 having equal performances under the SCSI-based scheme, R1 {\rm i.i.d.} performs better than R1 correlated case under the PCSI-based scheme. 
\begin{figure}[ht!]
\centering\vspace{-.4cm}
\begin{minipage}{.4\textwidth}
  \centering
  \includegraphics[width=\textwidth]{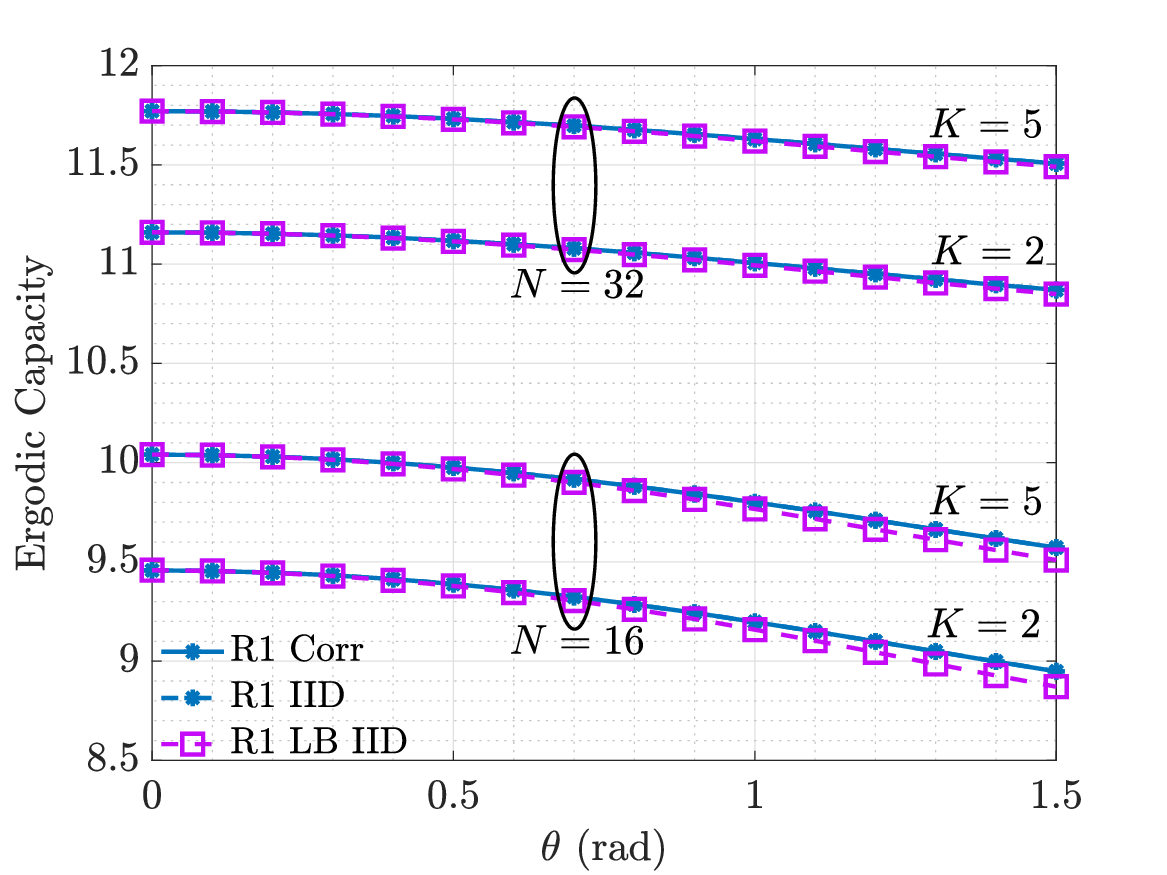}
\end{minipage}%
\begin{minipage}{.4\textwidth}
  \centering
  \includegraphics[width=\textwidth]{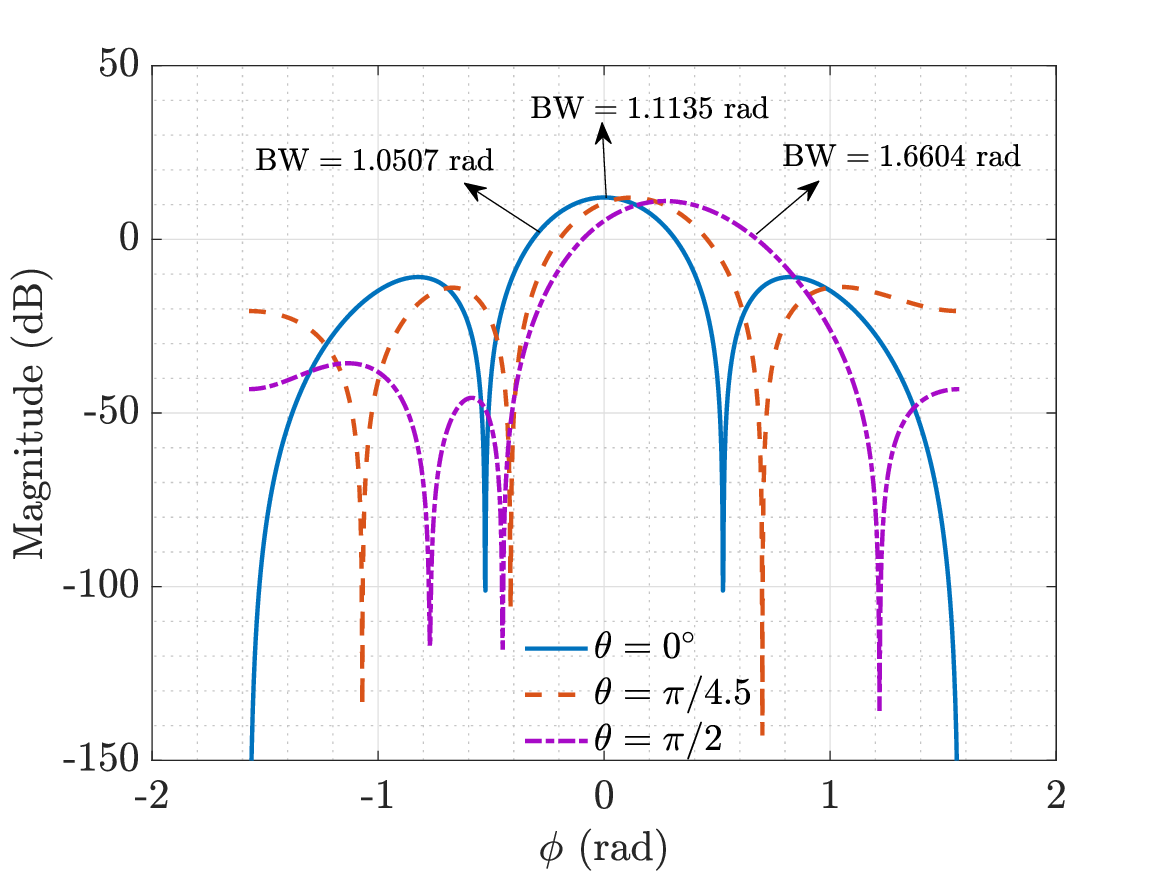}
\end{minipage}\vspace{-.2cm}
\caption{Left: EC vs. $\theta$ under R1 for correlated and {\rm i.i.d.} fading scenarios. Right: Radiation pattern for $\mathbf{f}_{\rm opt}$ with $N = 32$.}\vspace{-.6cm}
\label{Fig3}
\end{figure}

Fig. \ref{Fig3} (Left) shows that EC deteriorates  as $\theta$ increases where $\theta = |\theta_{\rm{DBR}} - \theta_{\rm{DBU}}|$ is the difference between AoDs of the direct and indirect links from the BS. This is because the optimal beamformer $\mathbf{f}$ forms a narrow lobe when $\theta$ is small. However, for larger $\theta$, the optimal beamformer $\mathbf{f}$ has a relatively wider beam 
to exploit the gains from both the direct and indirect links. This can also be observed from Fig. \ref{Fig3} (Right), where we plot the radiation pattern for an optimal transmit beamformer. The figure verifies that the main lobe beamwidth gradually increases in $\theta$ reducing the overall array gain. However, as the indirect link strength improves by increasing $N$, it can be seen that the drop in capacity with increasing $\theta$ becomes less severe. Further, it can  be observed that the EC LB is tight for any $\theta$, especially for large $N$. 
\begin{figure}[ht!]
\centering\vspace{-.4cm}
\begin{minipage}{.33\textwidth}
  \centering
  \includegraphics[width=\textwidth]{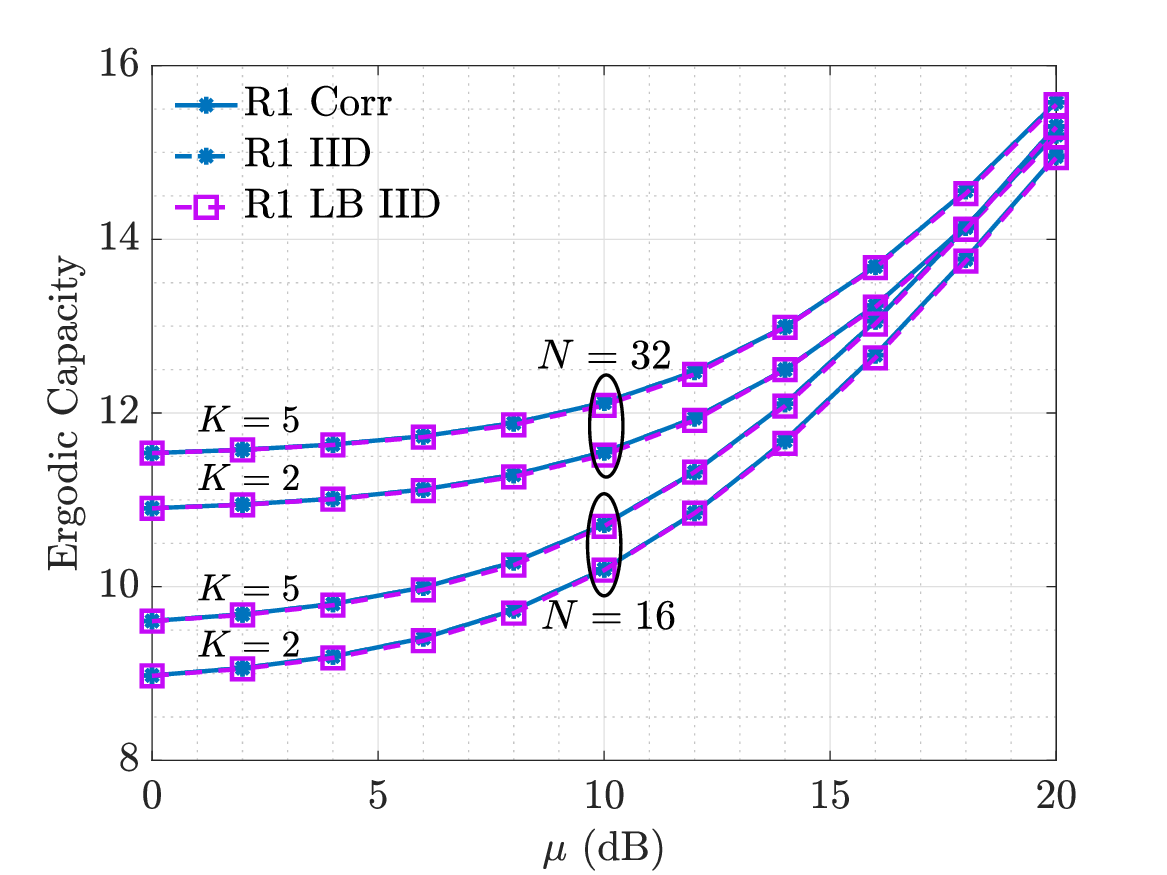}
\end{minipage}%
\begin{minipage}{.33\textwidth}
  \centering
  \includegraphics[width=\textwidth]{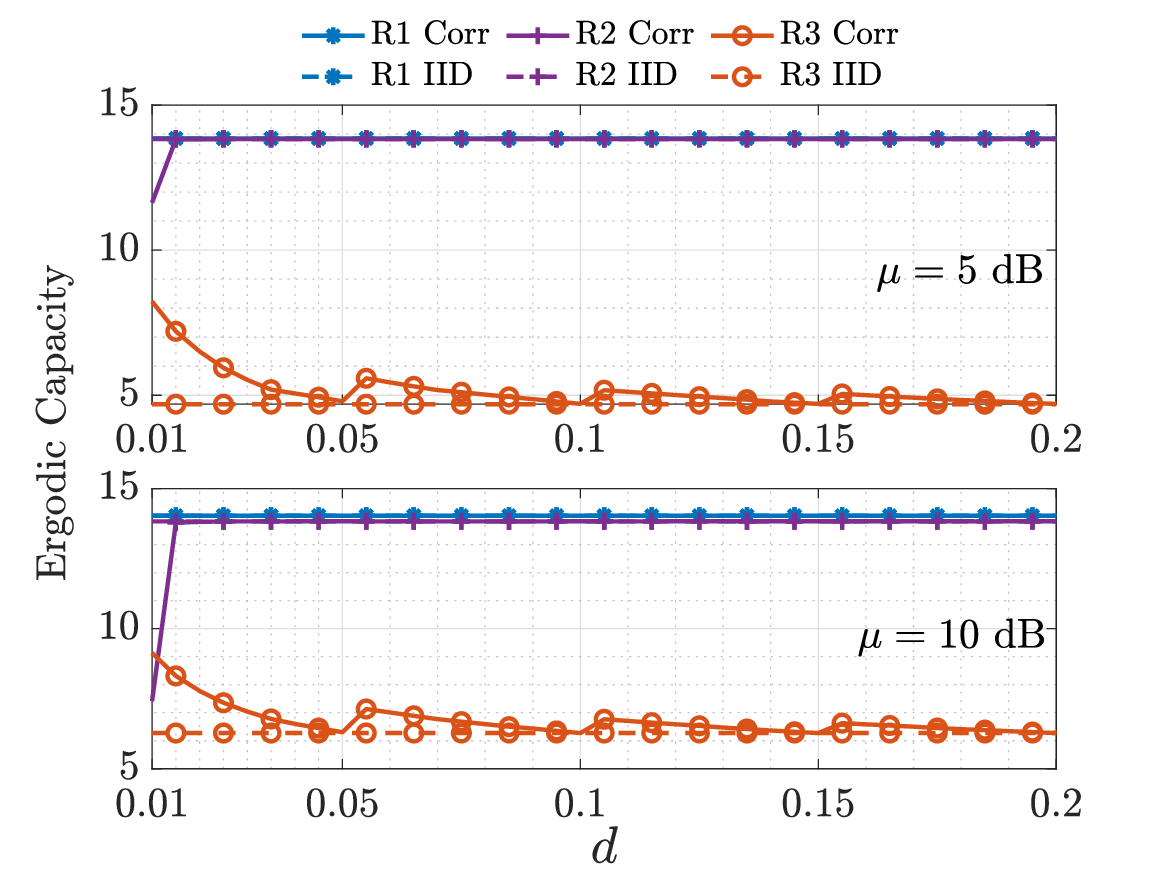}
\end{minipage}\vspace{-.3cm}
\begin{minipage}{.33\textwidth}
  \centering
  \includegraphics[width=\textwidth]{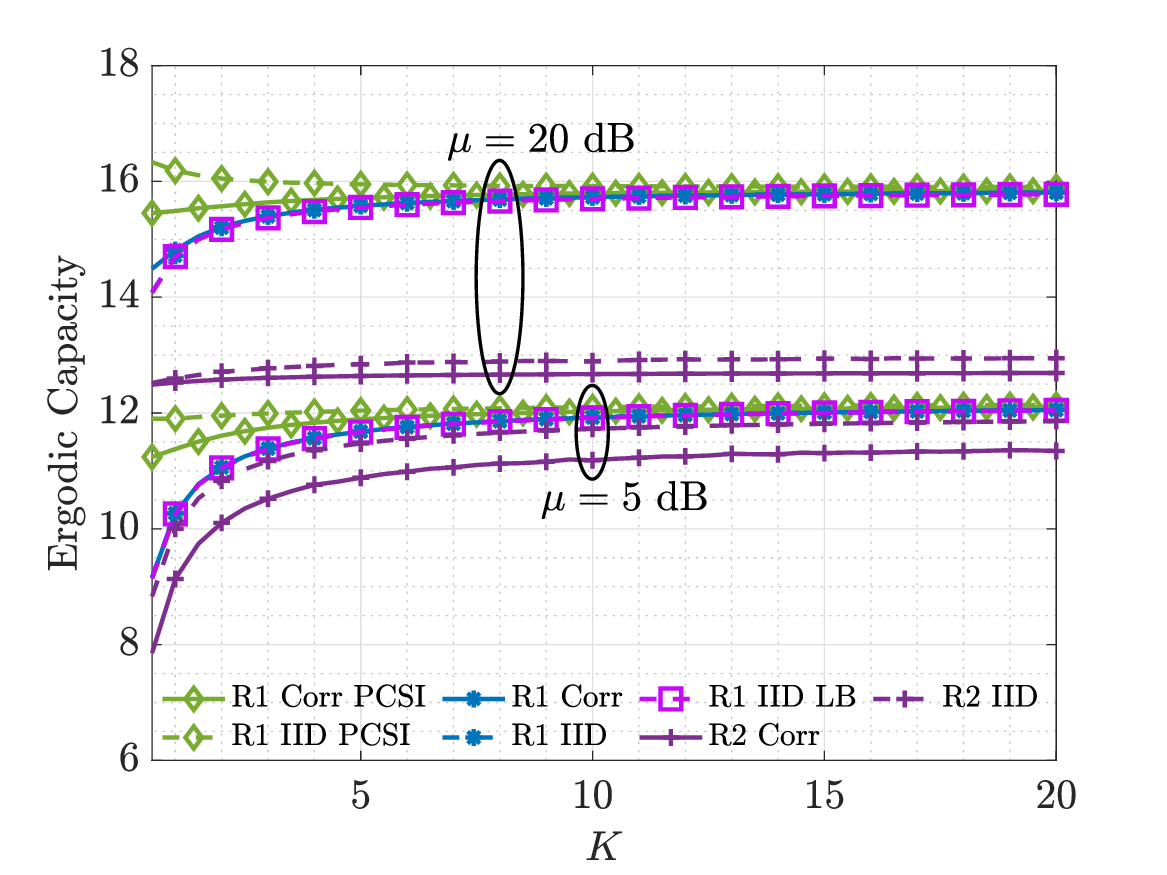}
\end{minipage}
\caption{Left: EC vs. the PLR $\mu$. Center: EC vs. distance between antennas $d$ for $M = 32$. Right: EC vs. Rician factor $K$.}\vspace{-.6cm}
\label{Fig4}
\end{figure}

Fig. \ref{Fig4} (Left) shows EC as a function of PLR $\mu$. It can be seen that the capacity increases with the increase in $\mu$. However, it can also be seen that the SCSI-based scheme performs equally when $\mu$ is large. This is mainly because the strong direct link (i.e., high $\mu$) provides enough capacity such that the presence of RIS cannot improve the mean SNR  further. Furthermore, it may be noted that the LB is accurate for a wide range of $\mu$.  
Fig. \ref{Fig4} (Center) shows the impact of the distance $d$ between consecutive antennas on EC performance under the SCSI-based scheme. Note that the correlation between the antenna elements depends on $d$ (refer to Section \ref{channel model}). Thus, Fig. \ref{Fig4} (Right) captures the impact of correlated fading on the performance of EC. First, we observe that EC under R1 correlated case remains constant as $d$ increases and is similar to the performance of the R1 {\rm i.i.d.} case. Moreover, the capacity under R3 correlated case reduces with increasing $d$ and finally converges to the capacity achieved under the R3 {\rm i.i.d.} case. This also verifies the fact that the correlated setting performs better than the {\rm i.i.d.} scenario under the R3 case, as pointed out before. Interestingly, we can observe that EC of the R3 correlated case gradually drops in a pattern with increase in $d$ and become exactly equal to capacity under the R3 {\rm i.i.d.} case when $d$ is equal to $\frac{\lambda}{2},~\lambda,~\frac{3\lambda}{2},~2 \lambda$. This is because the correlation matrix is constructed using a model based on the Sinc function which goes to zero for the above values of $d$ (meaning the correlated case reduces to  {\rm i.i.d.} for these $d$'s).
Fig. \ref{Fig4} (Right) shows that the EC increases with $K$. This is expected as $K$ indicates the strength of the LoS component that is crucial for achieving capacity under SCSI-based schemes. Furthermore, it can also be seen that the performances of the PCSI- and proposed SCSI-based schemes for the R1 case converge to the same constant value for large $K$. This is because the multi-path fading vanishes as $K \to \infty$, which implies the channel's deterministic nature for which PCSI- and SCSI-based beamforming schemes are the same. The figure also verifies the accuracy of LB EC  for a wide range of $K$.  
\vspace{-0.3cm}
\section{Conclusion}\vspace{-.2cm}
This paper investigates the statistically optimal transmit beamforming and phase shift matrix problem of a RIS-aided MISO  downlink communication system. We considered a generalized fading scenario wherein  both the direct and indirect links consist LoS paths  along with the multi-path components. Moreover, we  considered correlated fading to account for the practical deployment of RISs containing densely packed antenna elements. 
For this setting, we proposed an iterative algorithm  to optimally select the beamformers that  maximize the  mean SNR. In addition, we also derived approximately achievable OP and EC of the proposed algorithm. 
As is usually the case in the literature, this analysis relied on the numerical evaluations of the optimal beamforming solution. In order to get more crisp analytical insights, we further derived closed-form expressions for the optimal beamforming under the absence of LoS components and/or correlated fading along the direct and/or indirect links.     
Our analysis revealed some interesting interrelations between different fading scenarios.
For instance, we have shown that the maximum mean SNR improves linearly/quadratically with the number of RIS elements in the absence/presence of LoS component under {\rm i.i.d.} fading. 
Further, we have analytically shown that the statistically optimal beamforming performs better under correlated fading compared to {\rm i.i.d.} case in the absence of LoS paths. 
Our numerical results also show that the correlated fading is advantageous/disadvantageous compared to the  {\rm i.i.d.} case in the absence/presence of LoS paths. 
\vspace{-1cm}\appendix\vspace{-.5cm}
\subsection{Proof of \eqref{eq:Mean_SNR_R1}}\label{AppA}\vspace{-.3cm}
The mean {\rm SNR} for the channel model described in Section \ref{channel model} can be obtained as
\begin{align}
    &\Gamma(\mathbf{f}, \mathbf{\Phi})
    = \mathbb{E}[(\mathbf{h}^T\mathbf{\Phi Hf} + \mu \mathbf{g}^T\mathbf{f})(\mathbf{h}^T\mathbf{\Phi Hf} + \mu \mathbf{g}^T\mathbf{f})^H]\nonumber\\
     &\stackrel{(a)}{=} |\kappa_l^2\mathbf{\Bar{h}}^T\mathbf{\Phi}\mathbf{\Bar{H}f} + \kappa_l\mathbf{\Bar{g}}^T\mathbf{f}|^2+ \kappa_l^2\kappa_n^2 \mathbb{E}[|\mathbf{\Bar{h}}^T\mathbf{\Phi}\mathbf{\Tilde{H}f}|^2] + \kappa_l^2\kappa_n^2\mathbb{E}[|\mathbf{\Tilde{h}}^T\mathbf{\Phi}\mathbf{\Bar{H}f}|^2] + \kappa_n^4\mathbb{E}[|\mathbf{\Tilde{h}}^T\mathbf{\Phi}\mathbf{\Tilde{H}f}|^2].\label{ESNR}  
\end{align}
where step (a) is obtained by substituting $\mathbf{g}$, $\mathbf{h}$ and $\mathbf{H}$ from \eqref{directlink},\eqref{indirect link_h} and \eqref{indirect link_H} and simplifying. 
To  simplify further, we will use the following two identities.
    \newline$\bullet$ [I-1] For given $\mathbf{A}$ and $\mathbf{X}_{:,i} \in \mathcal{CN}(0, \mathbf{I})$, we have [Reference]
            $\mathbb{E}[\mathbf{X A}\mathbf{X}^H] = \rm{trace}\{\mathbf{A}\}\mathbf{I}.$ 
    \newline$\bullet$ [I-2] For given $\mathbf{A}$ and $\mathbf{B}$, we have
            $${\rm trace}\{\mathbf{\Phi}^H \mathbf{A} \mathbf{\Phi} \mathbf{B}\} \stackrel{(a)}{=} \sum\nolimits_{i,j}\boldsymbol{\psi}^*_i\boldsymbol{\psi}_j\mathbf{A}_{ij}\mathbf{B}_{ji}\stackrel{(b)}{=}\boldsymbol{\psi}^H(\mathbf{A}^T\odot \mathbf{B})\boldsymbol{\psi},$$
            where steps (a) and (b) follow using $\mathbf{\Phi}={\rm diag}(\boldsymbol{\psi})$ and $\boldsymbol{\psi}^H\mathbf{X}\boldsymbol{\psi}=\sum_{i,j}\boldsymbol{\psi}_i^*\mathbf{X}_{ij}\boldsymbol{\psi}_j$, respectively. 

Now, the second term in RHS of \eqref{ESNR} can be simplified as
\begingroup
\allowdisplaybreaks
\begin{align}
    \mathbb{E}[|\mathbf{\Bar{h}}^T\mathbf{\Phi}\mathbf{\Tilde{H}f}|^2] 
    &= \mathbb{E}[\mathbf{f}^H \Tilde{\mathbf{R}}_{\rm BT} \Tilde{\mathbf{H}}^H_W \Tilde{\mathbf{R}}_{\rm RR}\mathbf{\Phi}^H \Bar{\mathbf{h}}^* \Bar{\mathbf{h}}^T \mathbf{\Phi} \Tilde{\mathbf{R}}_{\rm RR} \Tilde{\mathbf{H}}_W \Tilde{\mathbf{R}}_{\rm BT}\mathbf{f}], \nonumber\\
     &\stackrel{(a)}{=}\mathbf{f}^H \Tilde{\mathbf{R}}_{\rm BT}{\rm trace}(\Tilde{\mathbf{R}}_{\rm RR}\mathbf{\Phi}^H \Bar{\mathbf{h}}^* \Bar{\mathbf{h}}^T \mathbf{\Phi} \Tilde{\mathbf{R}}_{\rm RR})  \mathbf{I}\Tilde{\mathbf{R}}_{\rm BT}\mathbf{f}, \nonumber\\
    &= \mathbf{f}^H\mathbf{R}_{\rm BT}\mathbf{f} \times {\rm trace}({\mathbf{R}}_{\rm RR}\mathbf{\Phi}^H \Bar{\mathbf{h}}^* \Bar{\mathbf{h}}^T \mathbf{\Phi}),\nonumber\\
    &\stackrel{(b)}{=}\mathbf{f}^H\mathbf{R}_{\rm BT}\mathbf{f} \times \boldsymbol{\psi}^H({\mathbf{R}}_{\rm RR} \odot \mathbf{\Bar{h}}^*\mathbf{\Bar{h}}^T)\boldsymbol{\psi}.\label{Exp1}
\end{align}
\endgroup
where   step (a) follows from the identity I-1, and the  step (b) from identity I-2. Similarly, using I-1 and I-2, we simplify the third and fourth terms in  the RHS of \eqref{ESNR} as below
\begin{align}
    \mathbb{E}[|\mathbf{\Tilde{h}}^T\mathbf{\Phi}\mathbf{\Bar{H}f}|^2] &=\boldsymbol{\psi}^H(\mathbf{R}_{\rm RT} \odot \mathbf{\Bar{H}f}\mathbf{f}^H\mathbf{\Bar{H}}^H)\boldsymbol{\psi},\label{Exp2}\\
   \text{~~and~~} \mathbb{E}[|\mathbf{\Tilde{h}}^T\mathbf{\Phi}\mathbf{\Tilde{H}f}|^2] 
    &=\mathbf{f}^H\mathbf{R}_{\rm BT}\mathbf{f} \times \boldsymbol{\psi}^H({\mathbf{R}}_{\rm RR} \odot \mathbf{R}_{\rm RT})\boldsymbol{\psi}.\label{Exp3}
\end{align}
Further, substituting \eqref{Exp1}, \eqref{Exp2}, and \eqref{Exp3} in \eqref{ESNR}, we obtain \eqref{eq:Mean_SNR_R1}.
\vspace{-.4cm}
\subsection{Distributions of $\xi_1$ and $\xi_2$}\label{AppB}
We first derive the means and variances of $\xi_1=\mathbf{h}^T\mathbf{\Phi}\mathbf{H f}$ and $\xi_2=\mathbf{g}^T\mathbf{f}$.
 From the definitions of $\mathbf{H}$, $\mathbf{h}$ and $\mathbf{g}$ given in Section \ref{channel model} and their independence, the means of $\xi_1$ and $\xi_2$ becomes
 \begin{align}
       \mu_1=\mathbb{E}[\xi_1] =  \mathbb{E}[\mathbf{h}^T\mathbf{\Phi}\mathbf{H f}]=\kappa_l^2 \Bar{\mathbf{h}} \mathbf{\Phi} \Bar{\mathbf{H}} \mathbf{f},~~\text{and}~~ \mu_2=\mathbb{E}[\xi_2] = \mathbb{E}[\mathbf{g}^T\mathbf{f}]=\kappa_l\mathbf{\bar{g}}^T\mathbf{f}.  \label{appendixB_mu1mu2}
 \end{align}
The variance of $\xi_1$  can be obtained as
 \begin{align}
       \sigma_1^2&=\mathbb{E}[\mathbf{h}^T\mathbf{\Phi}\mathbf{H}\mathbf{f}(\mathbf{h}^T\mathbf{\Phi}\mathbf{H}\mathbf{f})^H]-|\mu_1|^2,\nonumber\\
       &=  \kappa_l^2\kappa_n^2 \mathbb{E}[|\mathbf{\Bar{h}}^T\mathbf{\Phi}\mathbf{\Tilde{H}f}|^2] + \kappa_l^2\kappa_n^2\mathbb{E}[|\mathbf{\Tilde{h}}^T\mathbf{\Phi}\mathbf{\Bar{H}f}|^2] + \kappa_l^2\kappa_n^2\mathbb{E}[|\mathbf{\Tilde{h}}^T\mathbf{\Phi}\mathbf{\Tilde{H}f}|^2],\nonumber\\
       &=\kappa_l^2\kappa_n^2\boldsymbol{\psi}^H\mathbf{Z}_1 \boldsymbol{\psi} + \kappa_n^2\mathbf{f}^H\mathbf{R}_{\rm BT}\mathbf{f}[\boldsymbol{\psi}^H\mathbf{Z}_2\boldsymbol{\psi}],\label{appendixB_sigam_1}
 \end{align}
 where the last equality follows using \eqref{Exp1}, \eqref{Exp2} and \eqref{Exp3}. The variance of $\xi_2$ becomes 
 \begin{align}
     \sigma_2^2&=\mathbb{E}[\mathbf{f}^H\mathbf{{g}}^*\mathbf{{g}}^T\mathbf{f}]-|\mu_2|^2=\kappa_n^2\mathbf{f}^H\mathbf{R}_{\rm BT}\mathbf{f}.
 \end{align}
We now comment on the distributions of $\xi_1$ and $\xi_2$. Note that $\xi_2 = \kappa_l \Bar{\mathbf{g}}^T\mathbf{f} + \kappa_n \Tilde{\mathbf{g}}^T\mathbf{f}$ 
where $\Bar{\mathbf{g}}$ is deterministic and $\Tilde{\mathbf{g}} \sim \mathcal{CN}(0, \mathbf{R}_{\rm BT})$. It can be easily shown that $\Tilde{\mathbf{g}}^T\mathbf{f} \sim \mathcal{CN}(0, \mathbf{f}^H\mathbf{R}_{\rm BT}\mathbf{f})$. Thus,  $\xi_2$ becomes complex Gaussian with mean $\mu_2$ and variance $\sigma^2_2$ as given in \eqref{eq:xi1xi2}.
Next, $\xi_1=\mathbf{h}^T\mathbf{\Phi}{\mathbf{H}}\mathbf{f}$ can be expanded as 
$\xi_1  = \kappa_l^2 \Bar{\mathbf{h}}^T\mathbf{\Phi}\Bar{\mathbf{H}}\mathbf{f}  +  \kappa_l \kappa_n \Tilde{\mathbf{h}}^T\mathbf{\Phi}\Bar{\mathbf{H}}\mathbf{f}  +  \kappa_l \kappa_n \Bar{\mathbf{h}}^T\mathbf{\Phi}\Tilde{\mathbf{H}}\mathbf{f}  +  \kappa_n^2\Tilde{\mathbf{h}}^T\mathbf{\Phi}\Tilde{\mathbf{H}}\mathbf{f}$.
Here, the first term is deterministic, whereas the  second and third terms are complex Gaussian as they are linear combinations of elements of $\tilde{\mathbf{h}}$ and $\tilde{\mathbf{H}}$, respectively. 
However, the last term is the sum of products of two zero-mean complex Gaussian random variables as given by 
$$\Tilde{\mathbf{h}}^T\mathbf{\Phi}\Tilde{\mathbf{H}}\mathbf{f} = \sum\nolimits_{i = 1}^N \sum\nolimits_{j = 1}^M \Tilde{\mathbf{h}}_i \Tilde{\mathbf{H}}_{ij} \mathbf{\boldsymbol{\psi}}_i \mathbf{f}_j.$$
The exact distribution of the above form is challenging to derive  as the distribution of the product of two complex Gaussian random variables itself is in a complicated form \cite{Donoughue_GuassianProduct}, which naturally will lead to intractability in further analysis. However, as the above summation includes many terms (mainly when the number of RIS elements is large),  we can apply the central limit theorem to approximate its distribution  using a complex Gaussian. Thus, we can conclude that $\xi_2$ closely follows complex Gaussian distribution with mean $\mu_2$ and variance $\sigma^2_2$ as in \eqref{eq:xi1xi2}.
\vspace{-0.4cm}
\bibliographystyle{IEEEtran}
\bibliography{Reference}

\end{document}